\documentclass[11pt, cm, authoryear]{elegantpaper}

\usepackage{amsmath}
\usepackage{bm}       % 粗体表示向量、矩阵
\usepackage{array}    % 扩展了array和tabular环境功能，增强了列样式选项
\usepackage{makecell} % 可使表格中的数据单独定位(l,c,r)
\usepackage{multirow} % \multirow 可以在表格中排版横跨两行以上的文本
\usepackage{colortbl} % 表格颜色
\usepackage{diagbox}  % 斜线表头
\usepackage{xfrac}    % 更舒服的分数

\newcommand{\Ebb}{\mathbb{E}}

\newcommand{\Rbb}{\mathbb{R}}

\newcommand{\Zbb}{\mathbb{Z}}

\newcommand{\Ecal}{\mathcal{E}}
\newcommand{\Fcal}{\mathcal{F}}

\newcommand{\Ical}{\mathcal{I}}

\newcommand{\Mcal}{\mathcal{M}}

\newcommand{\Pcal}{\mathcal{P}}

\newcommand{\Tcal}{\mathcal{T}}

\newcommand{\st}{\mathrm{s.t.}}

\DeclareMathOperator{\diff}{d\!} % 微分符号
\DeclareMathOperator{\argmax}{argmax}

\usepackage{tikz}
\usepackage{tikz-cd}
\usetikzlibrary{
	calc,
	decorations.pathmorphing,
	matrix,arrows,
	positioning,
	shapes.geometric
}
\usepackage{pgfplots}
\pgfplotsset{compat=newest}

\newcommand{\fun}{f \colon 2^\Omega \to \Rbb_+}
\newcommand{\gfun}{g \colon 2^\Omega \to \Rbb_+}
\newcommand{\lfun}{\ell \colon 2^\Omega \to \Rbb_+}

\newcommand{\xvec}{\boldsymbol{x}}
\newcommand{\yvec}{\boldsymbol{y}}
\newcommand{\zvec}{\boldsymbol{z}}
\newcommand{\lvec}{\boldsymbol{\ell}}
\newcommand{\vvec}{\boldsymbol{v}}
\newcommand{\wvec}{\boldsymbol{w}}

\newcommand{\gset}{[0,1]^\Omega}
\newcommand{\ind}{\mathbf{1}}
\newcommand{\mcg}{\textsc{Measured Continuous Greedy}}
\newcommand{\drg}{\textsc{Distorted Random Greedy}}
\newcommand{\drsg}{\textsc{Distorted Random Sampling Greedy}}
\newcommand{\udrg}{\textsc{Unconstrained Distorted Random Greedy}}

\newcommand{\innprod}[2]{\langle #1,#2 \rangle}

\newcommand{\ai}{ \left( 1- \frac1k \right)^{k-i} }
\newcommand{\bi}{ \left( 1- \frac1k \right)^{k-(i+1)} }
\newcommand{\ain}{ \left( 1- \frac1n \right)^{n-i} }
\newcommand{\bin}{ \left( 1- \frac1n \right)^{n-(i+1)} }

\DeclareMathOperator{\pr}{Pr}

\usepackage[ruled,linesnumbered]{algorithm2e}
\title{Regularized Non-monotone Submodular Maximization}

\author{Cheng Lu \and Wenguo Yang\thanks{Email: yangwg@ucas.ac.cn} \and Suixiang Gao }% <-this % stops a space}
\institute{School of Mathematical Sciences, University of Chinese Academy of Sciences}

\date{\today}

\begin{document}

\maketitle

\begin{abstract}
In this paper, we present a thorough study of maximizing a regularized non-monotone submodular function subject to various constraints, i.e., $\max \{ g(A) - \ell(A) : A \in \Fcal \}$, where $\gfun$ is a non-monotone submodular function, $\lfun$ is a normalized modular function and $\Fcal$ is the constraint set. Though the objective function $f := g - \ell$ is still submodular, the fact that $f$ could potentially take on negative values prevents the existing methods for submodular maximization from providing a constant approximation ratio for the regularized submodular maximization problem. To overcome the obstacle, we propose several algorithms which can provide a relatively weak approximation guarantee for maximizing regularized non-monotone submodular functions. More specifically, we propose a continuous greedy algorithm for the relaxation of maximizing $g - \ell$ subject to a matroid constraint. Then, the pipage rounding procedure (\cite{vondrak2013symmetry}) can produce an integral solution $S$ such that $\Ebb [g(S) - \ell(S)] \geq e^{-1}g(OPT) - \ell(OPT) - O(\epsilon)$. Moreover, we present a much faster algorithm for maximizing $g - \ell$ subject to a cardinality constraint, which can output a solution $S$ with $\Ebb [g(S) - \ell(S)] \geq (e^{-1} - \epsilon) g(OPT) - \ell(OPT)$ using $O(\frac{n}{\epsilon^2} \ln \frac 1\epsilon)$ value oracle queries. We also consider the unconstrained maximization problem and give an algorithm which can return a solution $S$ with $\Ebb [g(S) - \ell(S)] \geq e^{-1} g(OPT) - \ell(OPT)$ using $O(n)$ value oracle queries.

\end{abstract}

\section{Introduction}
Submodular functions arise naturally from combinatorial optimization as several combinatorial functions turn out to be submodular. A few example of such functions include rank functions of matroids, cut functions of graphs and di-graphs, entropy functions and covering functions. Thus, combinatorial optimization problems with a submodular objective funtion are essentially submodular optimization problems. Unlike minimization of submodular functions which can be done in polynomial time (\cite{grotschel1981ellipsoid,schrijver2000combinatorial,iwata2001combinatorial}), submodular maximization problems are usually NP-hard, for many classes of submodular functions, such as weighted coverage (\cite{feige1998threshold}) or mutual information (\cite{krause2012near}).

In general, submodular maximization problem can be formulated as follows.
\begin{equation} \label{submodular maximization problem}
	\begin{alignedat}{1}
		\max \quad & g(A) \\ 
		\st \quad & A \in \Fcal \subseteq 2^\Omega,
	\end{alignedat}
\end{equation}
where $\gfun$ is a submodular funtion, and $\Fcal$ is the family of subsets of $\Omega$ obeying the constraint. This problem not only captures many combinatorial optimization problems including Max-$k$-Coverage (\cite{khuller1999budgeted}), Max-Bisection (\cite{frieze1997improved}) and Max-Cut (\cite{ageev20010}), but also has wide applications in viral marketing (\cite{kempe2003maximizing}), sensor placement (\cite{krause2008near}), and machine learning (\cite{krause2007near,lin2010multi}).

However, there are many problems intending to maximize a combination of a submodular function and a modular function, which can be formulated as follows.
\begin{equation} \label{regularized submodular maximization}
	\begin{alignedat}{1}
		\max \quad & g(A) - \ell(A) \\ 
		\st \quad & A \in \Fcal \subseteq 2^\Omega,
	\end{alignedat}
\end{equation}
where $\gfun$ is a submodular funtion, $\lfun$ is a modular function and $\Fcal$ is the constraint set. Problem~(\ref{regularized submodular maximization}) has various interpretations. For example, the modular function $\ell$ can represent  a penalty or a regularizer term to alleviate over-fitting in machine learning. Another typical scenario is when $g$ models the revenue associated to a particular feasible set $A$ and $\ell$ represents a cost associated to $A$. Then, problem~\ref{regularized submodular maximization} corresponds to maximizing profits.

\subsection{Related works}
Various studies have been devoted to submodular maximization problems. \cite{nemhauser1978analysis} proposed a greedy algorithm for solving the monotone submodular maximization problem with a matroid constraint, i.e.,
\begin{equation} \label{mono sub max matroid}
	\begin{alignedat}{1}
		\max \quad & g(A) \\ 
		\st \quad & A \in \Ical \subseteq 2^\Omega,
	\end{alignedat}
\end{equation}
where $\gfun$ is a monotone non-decreasing submodular funtion, and $\Ical$ is the family of independent sets of a matroid $\Mcal = (\Omega, \Ical)$. When $\Mcal$ is a uniform matroid (i.e., cardinality constraint), \cite{nemhauser1978analysis} showed that the greedy algorithm they proposed can achieve an $(1 - \sfrac 1e)$ approximation ratio. When $\Mcal$ is a general matroid, they proved the greedy algorithm can return an $\frac 12$-approximate solution. Later, \cite{nemhauser1978best} proved that the $(1 - \sfrac 1e)$ approximation ratio is tight for maximizing a monotone submodular function subject to a cardinality constraint.

Nevertheless, it remains an open problem for about three decades that whether there exists a polynomial-time algorithm which can achieve an $(1 - \sfrac 1e)$ approximation ratio for problem~(\ref{mono sub max matroid}) with an arbitrary matroid constraint. Finally, \cite{calinescu2011maximizing} gave an affirmative answer. They presented the \emph{continuous greedy algorithm} which can achieve the $(1 - \sfrac 1e)$ approximation guarantee for the more general matroid constraint. Specifically, their algorithm consists of two main components, i.e., a relaxation solver and a rounding procedure. Firstly, they approximately solve a relaxation of problem~(\ref{mono sub max matroid}), i.e.,
\begin{equation*} 
	\begin{alignedat}{1}
		\max \quad & G(\xvec) \\ 
		\st \quad & \xvec \in \Pcal \subseteq \gset ,
	\end{alignedat}
\end{equation*}
where $G$ is the multilinear extension of submodular function $\gfun$, and $\Pcal$ is the matroid polytope corresponding to matroid $\Mcal = (\Omega, \Ical)$. Then, they utilize a rounding procedure called \emph{pipage rounding} to obtain an integral solution, which do not lose anything in the objective compared with the fractional solution. 

The results mentioned above are all about submodular maximization with a monotone objective. However, when one considers submodular objectives which are not monotone, less is known. \cite{buchbinder2014submodular} proposed an $e^{-1}$-approximation algorithm called \emph{random greedy} for maximizing a non-monotone submodular function subject to a cardinality constraint. Furthermore, they also presented a more involved algorithm which can achieve an improved approximation ratio of $e^{-1}+0.004$ for the same problem. On the hardness side, \cite{gharan2011submodular} proved that there is no algorithm which can achieve an approximation ratio better than $0.491$ for the problem using a polynomial number of value oracle queries. By modifying the continuous greedy algorithm proposed by \cite{calinescu2011maximizing}, \cite{feldman2011unified} proposed an algorithm called \emph{measured continuous greedy} for maximizing a general submodular function subject to an arbitrary matroid. They showed that the measured continuous greedy can ahieve an approximation ratio of $e^{-1} - o(1)$ for the non-monotone case. On the inapproximability side, \cite{gharan2011submodular} also proved that there is no algorithm which can achieve an approximation ratio better than $0.478$ when the constraint is a partition matroid using a polynomial number of value oracle queries.

Another improtant non-monotone submodular maximization problem is the unconstrained submodular maximization, i.e.,
\begin{equation*}
	\max_{A \subseteq \Omega} \quad g(A),
\end{equation*}
where $\gfun$ is a non-monotone submodular funtion. \cite{buchbinder2015tight} proposed a randomized algorithm called \emph{double greedy} for solving this problem. They proved that the algorithm achieves an approximation ratio of $\sfrac 12$. On the hardness side, \cite{feige2011maximizing}  proved that no polynomial time algorithm for the
unconstrained non-monotone submodular maximization can have an approximation ratio of $\sfrac 12 + \epsilon$ for any
constant $\epsilon > 0$ in the value oracle model.

Very recently, much attention has been attracted to the regularized submodular maximization, i.e., problem~(\ref{regularized submodular maximization}). We notice that by the definition of submodularity, function $f := g - \ell$ is still a submodular function. The only difference between problem~(\ref{submodular maximization problem}) and problem~(\ref{regularized submodular maximization}) is that the objective in problem~(\ref{regularized submodular maximization}), i.e., $f = g - \ell$, could possibly take on negative values, while the objective in problem~(\ref{submodular maximization problem}) is non-negative. Actually, \cite{feige2011maximizing} have shown that for a  submodular function $f$ without any restrictions, verifying whether the maximum of function $f$ is greater than zero is NP-hard and requires exponentially many queries in the value oracle model. Thus, it is impossible to design an algorithm which can achieve an multiplicative approximation factor\footnote{Namely, the algorithm can produce a solution $S \in \Fcal$ such that $g(S) - \ell(S) \geq \alpha \cdot \max \{ g(A) - \ell(A) : A \in \Fcal \}$, where $\alpha$ is a constant in $(0,1)$.} for problem~(\ref{regularized submodular maximization}). A line of research has shown that in this case we should consider a weaker notion of approximation, i.e., finding a solution $S \in \Fcal$ such that $g(S) - \ell(S) \geq \alpha \cdot g(OPT) - \ell(OPT)$, where $OPT \in \argmax \{ g(A) - \ell(A) : A \in \Fcal \}$.

\cite{sviridenko2017optimal} are the first to study the submodular maxization with an objective which could possibly take on negative values, i.e.,
\begin{equation} \label{sub-mod matroid}
	\begin{alignedat}{1}
		\max \quad & g(A) + \ell(A) \\ 
		\st \quad & A \in \Ical ,
	\end{alignedat}
\end{equation} 
where $\gfun$ is a monotone submodular
function, $\ell \colon 2^\Omega \to \Rbb$ is a normalized modular function, and $\Ical$ is the family of independent sets of matroid $\Mcal = (\Omega, \Ical)$. They proposed a randomized polynomial-time algorithm which can return a feasible set $S$ such that $\Ebb_{S} [ g(S) + \ell(S) ] \geq (1 - 1/e) g(OPT) + \ell(OPT)$, where $OPT$ is the optimal solution. Since \cite{sviridenko2017optimal} utilized a sophisticated continuous greedy algorithm together with a time-consuming guessing step, their results are of mainly theoretical interest. Later, \cite{feldman2020guess} reconsided this problem and designed a subtle distorted continuous greedy algorithm which can bypass the guessing step. But the optimization of multilinear extension prevents its practical application.

\cite{harshaw2019submodular} studied a regularized maximization problem, i.e.,
\begin{equation} \label{weakly sub-mod cardinality}
	\begin{alignedat}{1}
		\max_{A \subseteq \Omega} \quad & g(A) - \ell(A) \\ 
		\st \quad & |A| \leq k ,
	\end{alignedat}
\end{equation}  
where $\gfun$ is a monotone $\gamma$-weakly submodular function and $\lfun$ is a normalized modular function. They introduced a subtle distorted objective function and showed that the distorted greedy algorithm they proposed can return a feasible solution $S$ with $g(S) - \ell(S) \geq (1 - e^{-\gamma}) g(OPT) - \ell(OPT)$,  where $OPT$ is the optimal solution. Moreover, they extended their results to the unconstrained setting, i.e.,
\begin{equation*} \label{unconstained weakly sub-mod}
	\max_{A \subseteq \Omega}\ g(A) - \ell(A),
\end{equation*}
where $\gfun$ is a monotone $\gamma$-weakly submodular function, and $\lfun$ is a non-negative modular function. In this unconstrained setting, they presented a much more efficient algorithm which can output a set $S$ such that $g(S) - \ell(S) \geq (1 - e^{-\gamma}) g(OPT) - \ell(OPT)$,  where $OPT$ is the optimal solution.

Recently, \cite{kazemi2020regularized} considered a special case of problem~(\ref{weakly sub-mod cardinality}), in which $\gfun$ is a monotone submodular function. They proposed an one-pass streaming algorithm which can return a solution $S$ such that $g(S) - \ell(S) \geq (\phi^{-2} - \epsilon) \cdot g(OPT) - \ell(OPT)$, where $\phi$ is the golden ratio (i.e., $\phi = \frac 12 (1 + \sqrt{5})$). Moreover, they proposed a distributed algorithm which can produce a solution $S$ with $\Ebb_{S}[ g(S) - \ell(S) ] \geq (1 - \epsilon) [ g(OPT) - \ell(OPT) ]$.

\subsection{Our contributions}
Our contributions are presented as follows.
\begin{itemize}
	\item We study the regularized non-monotone submodular maximization problem with an arbitrary matroid constraint, i.e.,
	\begin{equation*} 
		\begin{alignedat}{1}
			\max_{A \subseteq \Omega} \quad & g(A) - \ell(A) \\ 
			\st \quad & A \in \Ical ,
		\end{alignedat}
	\end{equation*}
	where $\gfun$ is a non-monotone submodular function, $\ell \colon 2^\Omega \to \Rbb_+$ is a normalized modular function, and $\Ical$ is the family of independent sets of matroid $\Mcal = (\Omega, \Ical)$.
	
	Firstly, we study the relaxation of the above optimization problem, i.e.,
	\begin{equation*} 
		\begin{alignedat}{1}
			\max \quad & G(\xvec) - L(\xvec) \\ 
			\st \quad & \xvec \in \Pcal \subseteq \gset ,
		\end{alignedat}
	\end{equation*}
	where $G$ is the multilinear extension of non-monotone submodular function $\gfun$, $L$ is the multilinear extension of modular function $\ell \colon 2^\Omega \to \Rbb_+$, and $\Pcal$ is a down-monotone and solvable polytope. We propose an algorithm, which is based on the measured continuous greedy (\cite{feldman2011unified}) and the distorted objective (\cite{feldman2020guess}), for solving this problem. We summarize the result in the following theorem.
	
	\begin{theorem}
		Let $\gfun$ be a non-monotone submodular function and $\lfun$ be a normalized modular function. Let $G$ and $L$ be the multilinear extension of $g$ and $\ell$, respectively. Let $\Pcal$ be a down-monotone and solvable polytope. There exists an algorithm that given a parameter $\epsilon \in (0,1)$ can produce a fractional solution $\xvec \in \Pcal$ such that with high probability 
		\begin{equation*}
			G( \xvec ) - L( \xvec ) \geq e^{-1} g(OPT) - \ell (OPT) - 5 \epsilon M,
		\end{equation*}
		where $M = \max \left\{ \max_{e\in \Omega} g(e|\varnothing), -\min_{e\in \Omega} g(e|\Omega - e) \right\} > 0$ and $OPT \in \argmax \{ g(A) - \ell(A) : A \subseteq 2^{\Omega}, \ind_A \in \Pcal \}$. Moreover, the algorithm performs $O ( \frac{n^5}{\epsilon^3} \ln \frac{n^3}{\epsilon^2} )$ value oracle queries.
	\end{theorem}
	
	When the polytope $\Pcal$ is a matroid polytope, we can utilize the pipage rounding procedure (\cite{vondrak2013symmetry}) to obtain an integral solution without lossing anything in the objective. We summarize the result in the following theorem. 
	
	\begin{theorem}
		There exists a polynomial time algorithm that given a non-monotone submodular function $\gfun$, a normalized modular function $\lfun$, a matroid $\Mcal = (\Omega, \Ical)$ and a parameter $\epsilon \in (0,1)$, with high probability returns a set $S \in \Ical$ obeying
		\begin{equation*}
			\Ebb [ g(S) - \ell (S) ] \geq e^{-1} g(OPT) - \ell (OPT) - 5 \epsilon M,
		\end{equation*}
		where $M = \max \left\{ \max_{e\in \Omega} g(e|\varnothing), -\min_{e\in \Omega} g(e|\Omega - e) \right\} > 0$ and $OPT \in \argmax \{ g(A) - \ell(A) : A \in \Ical \}$.
	\end{theorem}
	
	\item We study the regularized non-monotone submodular maximization problem with a cardinality constraint, i.e.,
	\begin{equation*} 
		\begin{alignedat}{1}
			\max_{A \subseteq \Omega} \quad & g(A) - \ell(A) \\ 
			\st \quad & |A| \leq k ,
		\end{alignedat}
	\end{equation*}
	where $\gfun$ is a non-monotone submodular function, $\ell \colon 2^\Omega \to \Rbb_+$ is a normalized modular function, and $k$ is the cardinality constraint.
	
	Since cardinality constraint is essentially a uniform matroid constraint, the maximization problem with a cardinality constraint is a special case of its counterpart with a matroid constraint. We propose a randomized algorithm using $O(nk)$ value oracle queries for solving this problem. We summarize the result in the following theorem.
	
	\begin{theorem}
		There exists a randomized algorithm that given a non-monotone submodular function $\gfun$, a normalized modular function $\lfun$, and a cardinality $k$, returns a feasible set $S \subseteq \Omega$ with
		\begin{equation*}
			\Ebb [ g(S) - \ell (S) ] \geq e^{-1} g(OPT) - \ell (OPT) ,
		\end{equation*}
		where $OPT \in \argmax \{ g(A) - \ell(A) : A \subseteq \Omega, |A| \leq k \}$. And, the algorithm performs $O(nk)$ value oracle queries.
	\end{theorem}
	
	In addition, by using random sampling (\cite{buchbinder2017comparing}), we propose a randomized algorithm using $O(\frac{n}{\epsilon^2} \ln \frac 1\epsilon)$ value oracle queries for solving the same problem. We summarize the result in the following theorem.
	
	\begin{theorem}
		There exists a randomized algorithm that given a non-monotone submodular funtion $\gfun$, a normalized modular function $\ell \colon 2^\Omega \to \Rbb_+$, and parameters $k$ and $\epsilon \in (0, e^{-1})$, returns a feasible solution $S \subseteq \Omega$ with
		\begin{equation*}
			\Ebb_S [g(S) - \ell(S)] \geq (e^{-1} - \epsilon) g(OPT) - \ell(OPT),
		\end{equation*}
		where $OPT \in \argmax \{ g(A) - \ell(A) : A \subseteq \Omega, |A| \leq k \}$. And, the algorithm performs $O( \frac{n}{\epsilon^2} \ln \frac 1\epsilon )$ value oracle queries.
	\end{theorem}

	\item We study the regularized non-monotone submodular maximization problem with no constraint, i.e.,
	\begin{equation*} 
			\max_{A \subseteq \Omega} \quad g(A) - \ell(A) 
	\end{equation*}
	where $\gfun$ is a non-monotone submodular function, $\ell \colon 2^\Omega \to \Rbb_+$ is a normalized modular function.
	
	As a special case of the cardinality constraint (i.e., $k=n$), we propose a randomized algorithm, which has the same performance guarantee in expectation but only requires $O(n)$ value oracle queries, for solving the unconstrained problem. We summarize the result in the following theorem.
	
	\begin{theorem}
		There exists a randomized algorithm that given a non-monotone submodular function $\gfun$ and a normalized modular function $\lfun$, returns a set $S \subseteq \Omega$ obeying
		\begin{equation*}
			\Ebb [ g(S) - \ell (S) ] \geq e^{-1} g(OPT) - \ell (OPT) ,
		\end{equation*}
		where $OPT \in \argmax \{ g(A) - \ell(A) : A \subseteq \Omega \}$. And, the algorithm performs $O(n)$ value oracle queries.
	\end{theorem}
\end{itemize}

\subsection{Organization}
The rest of this paper is organized as follows. In section~\ref{sec: preliminaries}, we give the preliminary definitions and lemmas which will be used throughout the paper. We study the regularized non-monotone maximization problem with a matroid constraint in section~\ref{sec: matroid constraint} and propose a continuous greedy algorithm for solving it. In section~\ref{sec: cardinality constraint}, we present two fast algorithms for the regularized optimization problem with a cardinality constraint. Then, in section~\ref{sec: unconstrained}, we study the unconstrained problem and propose a fast randomized algorithm for solving it. Finally, we conclude this paper in section~\ref{sec:conclusion}.

\section{Preliminaries} \label{sec: preliminaries}
In this section, we describe the notations, definitions and lemmas which we will use in this paper.

\subsection{Set Functions}
Given a set $A$ and an element $e$, we use $A+e$ and $A-e$ as shorthands for the expression $A \cup \{e\}$ and $A \backslash \{ e \}$, respectively.

Let $\Omega$ be a ground set of size $n$. A set function $f \colon 2^\Omega \to \Rbb$ is \emph{submodular} if and only if $f(A) + f(B) \geq f(A\cup B) + f(A \cap B)$ for any $A,B \subseteq \Omega$. Submodularity can equivalently be characterized in terms of \emph{marginal gains}, defined by $f(e|A) := f(A+e) - f(A)$. Then, $f$ is submodular if and only if $f(e|A) \geq f(e|B)$ for any $A \subseteq B \subset \Omega$ and any $e \in \Omega \backslash B$. We denote the marginal gain of a set $B$ to a set $A$ with respect to function $f$ by $f(B|A) := f(A \cup B) - f(A)$. 

We say that a set function $f$ is \emph{monotone non-decreasing} if and only if $f(A) \leq f(B)$ for any $A \subseteq B \subseteq \Omega$. An equivalent definition is that $f(e|A) \geq 0$ for any $A \subseteq \Omega$ and any $e \in \Omega$. Similarly, We say that a set function $f$ is \emph{monotone non-increasing} if and only if $f(A) \geq f(B)$ for any $A \subseteq B \subseteq \Omega$.

A set function $f \colon 2^\Omega \to \Rbb$ is said to be \emph{normalized} if $f(\varnothing) = 0$. As a special case of submodular functions, a set function $f \colon 2^\Omega \to \Rbb$ is \emph{modular} if and only if $f(A) + f(B) = f(A\cup B) + f(A \cap B)$ for any $A,B \subseteq \Omega$. Based on the definitions above, one can easily check that a normalized set function $\ell \colon 2^\Omega \to \Rbb$ is modular if and only if there exists a vector $d \in \Rbb^{\Omega}$ such that $\ell(A) = \sum_{e \in A} d_e$ for every $A \subseteq \Omega$. Thus, in this paper, we abuse notation and identify the normalized modular function $\ell$ with the vector it corresponds to. Namely, $\ell(A) = \sum_{e \in A} \ell_e$ for every $A \subseteq \Omega$.

\subsection{Operations on Vectors}
Given two vectors $\xvec, \yvec \in \gset$, we write $\xvec \leq \yvec$ if $x_e \leq y_e$ for every element $e \in \Omega$. We use $\xvec \vee \yvec$, $\xvec \wedge \yvec$ and $\xvec \circ \yvec$ to denote the coordinate-wise maximum, minimum and multiplication of vectors $\xvec$ and $\yvec$, respectively. Specifically, $(\xvec \vee \yvec)_e := \max \{ x_e, y_e \}$, $(\xvec \wedge \yvec)_e := \min \{ x_e, y_e \}$ and $(\xvec \circ \yvec)_e := x_e \cdot y_e$, for every element $e \in \Omega$.

\subsection{Multilinear Extension}
Given a vector $\xvec \in \gset$, let $R_{\xvec}$ denote a random subset of $\Omega$ containing every element $e \in \Omega$ independently with probability $x_e$. Then, the \emph{multilinear extension} $F \colon \gset \to \Rbb$ of a set function $f \colon 2^\Omega \to \Rbb$ is defined as 
\begin{equation*}
	F(\xvec) := \Ebb_{ R_{\xvec} } [ f( R_{\xvec} ) ] = \sum_{S \subseteq \Omega} f(S) \prod_{e \in S} x_e \prod_{e \in \Omega - S} (1 - x_e). 
\end{equation*}
If we denote by $\ind_S$ the characteristic vector of set $S \subseteq \Omega$, then it holds that $F(\ind_S) = f(S)$. Thus, function $F$ is indeed an extension of function $f$. 

By the definition of multilinear extension, one can get that for every $e \in \Omega$ and every $\xvec \in \gset$,
\begin{equation*}
	(1 - x_e) \frac{\partial F(\xvec)}{\partial x_e} = F(\xvec \vee \ind_{ \{e\} }) - F(\xvec) = \Ebb_{ R_{\xvec} } [ f(e| R_{\xvec} ) ],
\end{equation*}
and,
\begin{equation*}
	\frac{\partial F(\xvec)}{\partial x_e} = F(\xvec \vee \ind_{ \{e\} }) - F(\xvec \wedge \ind_{\Omega - e}) = \Ebb_{ R_{\xvec} } [ f(e| R_{\xvec} - e ) ].
\end{equation*}
In addition, for every $e,e' \in \Omega$ and every $\xvec \in \gset$, it holds that
\begin{align*}
	\frac{\partial^2 F( \xvec )}{\partial x_e \partial x_e' } &= \Ebb_{ R_{\xvec} } [ f(e' | R_{\xvec} + e - e' ) - f(e' | R_{\xvec} - e - e' ) ].
\end{align*}

We now consider the multilinear extension of normalized modular functions. Suppose that $\ell \colon 2^\Omega \to \Rbb$ is a normalized modular function, then its multilinear extension is $L(\xvec) = \Ebb_{ R_{\xvec} } [ \ell(R_{\xvec}) ] = \sum_{e \in \Omega} \ell_e \cdot x_e = \innprod{\lvec}{\xvec}$ for every $\xvec \in \gset$.

\subsection{Lov\'{a}sz Extension}
Given a vector $\xvec \in \gset$ and a scalar $\lambda \in [0,1]$, let $T_{\lambda} (\xvec) := \{ e \in \Omega : x_e \geq \lambda \}$ be the set of elements in ground set $\Omega$ whose coordinate in $\xvec$ is at least $\lambda$. Then, the \emph{Lov\'{a}sz extension} $\hat{f} \colon \gset \to \Rbb$ of a submodular function $f \colon 2^\Omega \to \Rbb$ is defined as 
\begin{equation*}
	\hat{f}(\xvec) := \Ebb_{\lambda \sim U[0,1]} [ f( T_{\lambda} (\xvec) ) ] = \int_{0}^{1} f( T_{\lambda} (\xvec) ) \diff \lambda.
\end{equation*}

In this paper, we make use of the Lov\'{a}sz extension to lower bound the multilinear extension via the following lemma.

\begin{lemma}[Lemma A.4 in \cite{vondrak2013symmetry}] \label{lb multi}
	Let $F(\xvec)$ and $\hat{f}(\xvec)$ be the multilinear and Lov\'{a}sz extensions, respectively, of a submodular function $f \colon 2^\Omega \to \Rbb$. Then, it holds that $F(\xvec) \geq \hat{f}(\xvec)$ for every $\xvec \in \gset$.
\end{lemma}

\subsection{Polytopes}
A polytope $\Pcal \subseteq \gset$ is said to be \emph{down-monotone} if $\xvec \in \Pcal$ and $\mathbf{0} \leq \yvec \leq \xvec$ imply that $\yvec \in \Pcal$. A polytope $\Pcal$ is \emph{solvable} if there is an oracle for optimizing normalized modular functions over $\Pcal$, i.e., for solving $\max_{\xvec \in \Pcal} \innprod{\lvec}{\xvec}$ for any vector $\lvec$.

Given a matroid $\Mcal = (\Omega, \Ical)$, then its \emph{matroid polytope} $\Pcal(\Mcal)$ is
\begin{equation*}
	\Pcal(\Mcal) := \mathrm{conv} \{ \ind_I : I \in \Ical \} = \left\{ \xvec \geq \mathbf{0} : \sum_{e \in S} x_e \leq r_{\Mcal} (S), \forall S \subseteq \Omega \right\},
\end{equation*}
where $r_{\Mcal}$ is the rank function of matroid $\Mcal$. According to the definition, matroid polytope $\Pcal(\Mcal)$ is down-monotone and solvable. 

\subsection{Extra Assumptions}
In this paper, we study the regularized non-monotone submodular maximization problem, i.e.,
\begin{equation*} 
	\begin{alignedat}{1}
		\max \quad & g(A) - \ell(A) \\ 
		\st \quad & A \in \Ical \subseteq 2^\Omega ,
	\end{alignedat}
\end{equation*}
where $\gfun$ is a non-monotone submodular function, $\ell \colon 2^\Omega \to \Rbb_+$ is a normalized modular function, and $\Ical$ is the family of independent sets of a matroid $\Mcal = (\Omega, \Ical)$.

Here, $g$ is non-monotone means that function $g$ is neither monotone non-decreasing nor monotone non-increasing. When $g$ is monotone non-decreasing, the corresponding regularized submodular maximization problem has been studyed by \cite{sviridenko2017optimal}, \cite{feldman2020guess}, \cite{harshaw2019submodular} and \cite{kazemi2020regularized}. When $g$ is monotone non-increasing, the corresponding regularized submodular maximization problem is trivial, since $\varnothing$ is the optimal solution.

Because function $g$ is neither monotone non-decreasing nor monotone non-increasing, it holds that $\max_{e\in \Omega} g(e|\varnothing) > 0$ and $\min_{e\in \Omega} g(e|\Omega - e) < 0$. Hence, we have
\begin{equation*}
	M := \max \left\{ \max_{e\in \Omega} g(e|\varnothing), -\min_{e\in \Omega} g(e|\Omega - e) \right\} > 0.
\end{equation*}
Note that given a non-monotone submodular function $g \colon 2^\Omega \to \Rbb$, parameter $M$ is linear-time computable.

\section{Matroid Constraint} \label{sec: matroid constraint}
In this section, we propose an algorithm, which is presented as Algorithm~\ref{mcg distorted}, for solving the relaxed problem, i.e.,
\begin{equation*} 
	\begin{alignedat}{1}
		\max \quad & G(\xvec) - L(\xvec) \\ 
		\st \quad & \xvec \in \Pcal \subseteq \gset ,
	\end{alignedat}
\end{equation*}
where $G$ is the multilinear extension of non-monotone submodular function $\gfun$, $L$ is the multilinear extension of modular function $\ell \colon 2^\Omega \to \Rbb_+$, and $\Pcal$ is a down-monotone and solvable polytope.

The key techniques we use are the measured continuous greedy algorithm proposed by \cite{feldman2011unified}, and the distorted objective introduced by \cite{feldman2020guess}. Note that while \cite{feldman2020guess} chose $(1+\delta)^{\frac{t-1}{\delta}}$ as the coefficient of their distorted objective, we set the coefficient as $(1-\delta)^{\frac{t-1}{\delta} - 1}$, which is more appropriate for the problem we study. 

\begin{algorithm}[htbp]
	\caption{\mcg\ with distorted objective}
	\label{mcg distorted}
	\SetAlgoLined
	\KwIn{function $g$ and $\ell$, polytope $\Pcal$, error $\epsilon \in (0,1)$}
	\KwOut{approximate fractional solution $\yvec (1)$}
	Initialize $\yvec(0) \leftarrow \ind_{\varnothing}$, $t \leftarrow 0$, $\delta \leftarrow \lceil 2 + n^2 / \epsilon \rceil ^{-1}$\;
	\While{ $t < 1$ }{
		\ForEach{ $e \in \Omega$ }{
			Let $w_e(t)$ be an estimate for $\Ebb [ g(e|R_{\yvec(t)}) ]$ obtained by averaging the value of the expression within this expectation for $ r = \lceil 2n^2 \epsilon^{-2} \ln(2n \epsilon^{-1} \delta^{-1}) \rceil $ independent samples of $R_{\yvec(t)}$.
		}	
		$\zvec(t) \leftarrow \argmax \{ \innprod{\vvec}{ (1 - \delta)^{ \frac{1-t}{\delta} - 1 } \wvec(t) - \lvec } : \vvec \in \Pcal \}$ \;
		$ \yvec(t+\delta) \leftarrow \yvec(t) + \delta \cdot \zvec(t) \circ ( \ind_{\Omega} - \yvec(t) ) $ \;
		$t \leftarrow t + \delta$ \;
	}
	\Return $\yvec(1)$ \;
\end{algorithm}

According to the settings of Algorithm~\ref{mcg distorted}, we have $\delta \leq \min \{ \sfrac 12, \epsilon n^{-2} \} $.

\subsection{Feasibility of Solution}
Let $\Tcal$ be the set of times considered by Algorithm~\ref{mcg distorted}, i.e., $\Tcal := \{ i\delta : 0 \leq i < \delta^{-1}, i \in \Zbb \}$. Firstly, we prove that each coordinate of $\yvec(t)$ during the iteration has a non-trivial upper bound.

\begin{lemma} \label{y(t) bound}
	For every time $t \in \Tcal \cup \{1\}$ and every element $e \in \Omega$, it holds that $0 \leq y_e (t) \leq 1 - (1-\delta)^{\frac t\delta}$.
\end{lemma}

\begin{proof}
	We prove this lemma by induction on $t$.
	
	\noindent \textbf{Basis step:} Since $\yvec(0) = \ind_{\varnothing}$, the inequalities hold for $t = 0$.
	
	\noindent \textbf{Induction step:} Assume that $0 \leq y_e (t) \leq 1 - (1-\delta)^{\frac t\delta}$ is valid for time $t \in \Tcal$ and element $e \in \Omega$. We now prove that $0 \leq y_e (t+\delta) \leq 1 - (1-\delta)^{\frac t\delta + 1}$.
	
	According to Algorithm~\ref{mcg distorted}, we have $y_e (t+\delta) = y_e(t) + \delta z_e(t) ( 1 - y_e(t) )$. Moreover, $\zvec(t) \in \Pcal \subseteq \gset$ implies that $z_e(t) \in [0,1]$ for every $e \in \Omega$. Thus, it follows that $y_e (t+\delta) \geq 0$.
	
	Meanwhile, by assumption, we have
	\begin{align*}
		y_e (t+\delta) &\leq y_e(t) + \delta ( 1 - y_e(t) ) \\
		&\leq \delta + (1-\delta) \left[ 1 - (1-\delta)^{\frac t\delta} \right] \\
		&= 1 - (1-\delta)^{\frac t\delta + 1}.
	\end{align*}

	Hence, by the principle of induction, this lemma holds.
\end{proof}

Next, we prove that the solution which Algorithm~\ref{mcg distorted} produces is a feasible solution.

\begin{corollary}
	$\yvec(1) \in \Pcal$.
\end{corollary}

\begin{proof}
	According to Algorithm~\ref{mcg distorted}, we have $\yvec(1) = \sum_{t \in \Tcal} \delta \cdot \zvec(t) \circ ( \ind_{\Omega} - \yvec(t) )$. By Lemma~\ref{y(t) bound}, it holds that $\ind_{\varnothing} \leq \ind_{\Omega} - \yvec(t) \leq \ind_{\Omega}$ for every time $t \in \Tcal$, which indicates that $\ind_{\varnothing} \leq \zvec(t) \circ ( \ind_{\Omega} - \yvec(t) ) \leq \zvec(t)$ is valid for every time $t \in \Tcal$. Since $\Pcal$ is a down-monotone polytope and $\zvec \in \Pcal$, we get that $\zvec(t) \circ ( \ind_{\Omega} - \yvec(t) ) \in \Pcal$ for every time $t \in \Tcal$. 
	
	Thus, $\yvec(1) = \sum_{t \in \Tcal} \delta \cdot \zvec(t) \circ ( \ind_{\Omega} - \yvec(t) )$ is a convex combination of vectors in $\Pcal$. The convexity of $\Pcal$ implies that $\yvec(1) \in \Pcal$.
\end{proof}

\subsection{Good Estimator}
In this subsection, we prove that it's a low probability event that any of the estimates for $\Ebb [ g(e|R_{\yvec(t)}) ]$ made by Algorithm~\ref{mcg distorted} has a significant error. We will need the following lemma to assist our proof.

\begin{lemma}[Chernoff bound in  \cite{alon2004probabilistic}] \label{chernoff bound}
	Let $X_i (i = 1,2,\ldots,k)$ be mutually independent random variables with $\Ebb[X_i] = 0$ and $|X_i| \leq 1$ for every $i = 1,2,\ldots,k$. Set $S = X_1 + X_2 + \cdots + X_k$ and let $a$ be a positive number. Then, it holds that $\pr [ |S| > a ] \leq 2 e^{-a^2 / 2k}$.
\end{lemma}

Let $\Ecal(e,t)$ be the event that $\left|w_e(t) - \Ebb [ g(e|R_{\yvec(t)}) ] \right| \leq \frac{2 \epsilon M}{n}$ for element $e \in \Omega$ and time $t \in \Tcal$. And, let $\Ecal$ be the event $\cap_{e \in \Omega} \cap_{t \in \Tcal} \Ecal(e,t)$, namely, $\left|w_e(t) - \Ebb [ g(e|R_{\yvec(t)}) ] \right| \leq \frac{2 \epsilon M}{n}$ for every element $e \in \Omega$ and every time $t \in \Tcal$. Next, we show that $\Ecal$ is a high probability event.

\begin{lemma} \label{high P event}
	$\pr [\Ecal] \geq 1 - \epsilon $.
\end{lemma}

\begin{proof}
	Consider an arbitrary element $e \in \Omega$ and time $t \in \Tcal$, and let us denote by $R_i$ the $i$-th independent sample of $R_{\yvec(t)}$ used for calculating $w_e(t)$. We define random variables $X_i = \frac{1}{2M} \left\{ g(e|R_i) - \Ebb [ g(e|R_{\yvec(t)}) ] \right\} $ for every $i = 1,2,\ldots,r$.
	
	By the linearity of expectation, we have $\Ebb [X_i] = 0$. Since function $g$ is submodular, it holds that $g(e|R_i)$ and $\Ebb [ g(e|R_{\yvec(t)}) ]$ are both located in the interval $[-M,M]$. Thus, $X_i \in [-1,1]$ holds for every $i = 1,2,\ldots,r$. 
	
	By Lemma~\ref{chernoff bound}, it holds that
	\begin{align*}
		\pr \left[ \overline{ \Ecal(e,t) } \right] &= \pr \left[ \left|w_e(t) - \Ebb [ g(e|R_{\yvec(t)}) ] \right| > \frac{2 \epsilon M}{n} \right] \\
		&= \pr \left[ \left| \sum_{i=1}^{r} X_i \right| > \frac{\epsilon r}{n} \right] \\
		&\leq 2e^{-\frac{1}{2r} \left( \frac{\epsilon r}{n} \right)^2 }\\
		&= \frac{\epsilon \delta}{n}.
	\end{align*}

	By union bound, we have
	\begin{align*}
		\pr \left[ \overline{\Ecal} \right] &= \pr \left[ \cup_{e \in \Omega} \cup_{t \in \Tcal} \overline{ \Ecal(e,t) } \right] \\
		&\leq \sum_{e \in \Omega} \sum_{t \in \Tcal} \pr \left[ \overline{ \Ecal(e,t) } \right] \\
		&\leq \epsilon.
	\end{align*}
	Thus, $\pr [ \Ecal ] = 1 - \pr \left[ \overline{\Ecal} \right] \geq 1 - \epsilon$, which concludes the proof.
\end{proof}

\subsection{A Technical Lemma}
In this subsection, we prove a technical lemma which characterize the behavior of multilinear extension within a small neighbor.

\begin{lemma} \label{technical lemma}
	Given two vectors $\yvec,\yvec' \in \gset$ such that $\left| y'_e - y_e \right| \leq \delta \leq 1$ and a non-negative submodular function $\fun$ whose multilinear extension is $F$. Then, it holds that
	\begin{equation*}
		| F(\yvec') - F(\yvec) - \innprod{ \nabla F(\yvec) }{ \yvec' - \yvec } | \leq n^2 \delta^2 M,
	\end{equation*}
	where $n = |\Omega|$ and $M = \max \left\{ \max_{e\in \Omega} f(e|\varnothing), -\min_{e\in \Omega} f(e|\Omega - e) \right\} > 0$.
\end{lemma}

\begin{proof}
	According to Taylor's theorem, we have
	\begin{equation*}
		F(\yvec') - F(\yvec) = \innprod{\nabla F(\yvec)}{\yvec' - \yvec} + \frac 12 \sum_{e,e' \in \Omega} \frac{\partial^2 F( \yvec + t(\yvec'-\yvec) )}{\partial x_e \partial x_e' } \cdot (y'_{e} - y_{e})(y'_{e'} - y_{e'}),
	\end{equation*}
	where $t \in [0,1]$.
	
	For every $e,e' \in \Omega$ and every $\xvec \in \gset$, it holds that
	\begin{align*}
		\frac{\partial^2 F( \xvec )}{\partial x_e \partial x_e' } &= \Ebb_{ R_{\xvec} } [ f(e' | R_{\xvec} + e - e' ) - f(e' | R_{\xvec} - e - e' ) ].
	\end{align*}
	Meanwhile, by the definition of $M$, we have $f(e' | R_{\xvec} + e - e' ) , f(e' | R_{\xvec} - e - e' ) \in [-M,M]$. It follows that
	\begin{align*}
		\left| \frac{\partial^2 F( \xvec )}{\partial x_e \partial x_e' } \right| &= \left| \Ebb_{ R_{\xvec} } [ f(e' | R_{\xvec} + e - e' ) - f(e' | R_{\xvec} - e - e' ) ] \right| \\
		&\leq  \Ebb_{ R_{\xvec} } [ | f(e' | R_{\xvec} + e - e' ) | + | f(e' | R_{\xvec} - e - e' ) | ] \\
		&\leq 2M.
	\end{align*}

	Thus, we obtain
	\begin{align*}
		\left| F(\yvec') - F(\yvec) - \innprod{\nabla F(\yvec)}{\yvec' - \yvec} \right| &= \left| \frac 12 \sum_{e,e' \in \Omega} \frac{\partial^2 F( \yvec + t(\yvec'-\yvec) )}{\partial x_e \partial x_e' } \cdot (y'_{e} - y_{e})(y'_{e'} - y_{e'}) \right| \\
		&\leq \frac 12 \sum_{e,e' \in \Omega} \left| \frac{\partial^2 F( \yvec + t(\yvec'-\yvec) )}{\partial x_e \partial x_e' } \right| \cdot \left|(y'_{e} - y_{e})(y'_{e'} - y_{e'}) \right| \\
		&\leq n^2 \delta^2 M,
	\end{align*}
	which concludes the proof.
\end{proof}

\subsection{Performance Guarantee}
In this subsection, we analyze the performance guarantee of Algorithm~\ref{mcg distorted}. Firstly, we lower bound the increase of the multilinear extension $G$ of submodular function $g$ in each iteration.

\begin{lemma} \label{G bound}
	If the event $\Ecal$ happens, then, for every time $t \in \Tcal$, it holds that
	\begin{equation*}
		G( \yvec (t+\delta) ) - G( \yvec (t) ) \geq \delta \innprod{\wvec(t)}{\zvec(t)} - 3\epsilon \delta M,
	\end{equation*}
	where $M = \max \left\{ \max_{e\in \Omega} g(e|\varnothing), -\min_{e\in \Omega} g(e|\Omega - e) \right\} > 0$.
\end{lemma}

\begin{proof}
	By the settings of Algorithm~\ref{mcg distorted}, we have $\yvec(t+\delta) = \yvec(t) + \delta \cdot \zvec(t) \circ ( \ind_{\Omega} - \yvec(t) )$. According to Lemma~\ref{y(t) bound} and $\zvec(t) \in \Pcal \subseteq \gset$, it holds that $\left| y_e (t+\delta) - y_e(t) \right| = \left| \delta z_e(t) ( 1 - y_e(t) ) \right| \leq \delta$ for every element $e \in \Omega$. 
	
	Thus, we get
	\begin{align*}
		G( \yvec (t+\delta) ) - G( \yvec (t) ) &\geq \innprod{\nabla G( \yvec ) }{ \yvec (t+\delta) - \yvec(t) } - n^2 \delta^2 M \\
		&= \innprod{\nabla G( \yvec ) }{ \delta \cdot \zvec(t) \circ ( \ind_{\Omega} - \yvec(t) ) } - n^2 \delta^2 M \\
		&= \delta \cdot \sum_{e \in \Omega} z_e(t) \cdot \Ebb[ g(e| R_{\yvec(t)} ) ] - n^2 \delta^2 M,
	\end{align*}
	where the first inequality follows by Lemma~\ref{technical lemma}, the first equality follows by $\yvec(t+\delta) = \yvec(t) + \delta \cdot \zvec(t) \circ ( \ind_{\Omega} - \yvec(t) )$, and the second equality follows by the property of partial derivatives of multilinear extension and the definition of operation $\circ$ and inner product.
	
	When the event $\Ecal$ happens, it holds that $\Ebb[ g(e| R_{\yvec(t)} ) ] \geq w_e(t) - \frac{2\epsilon M}{n}$ for every element $e \in \Omega$ and every time $t \in \Tcal$. Then, it follows that
	\begin{align*}
		G( \yvec (t+\delta) ) - G( \yvec (t) ) &\geq \delta \cdot \sum_{e \in \Omega} z_e(t) \cdot \left[ w_e(t) - \frac{2\epsilon M}{n} \right] - n^2 \delta^2 M \\
		&\geq \delta \innprod{\wvec(t)}{\zvec(t)} - 2\epsilon \delta M - n^2 \delta^2 M \\
		&\geq \delta \innprod{\wvec(t)}{\zvec(t)} - 3\epsilon \delta M,
	\end{align*}
	where the last inequality follows by $\delta \leq \min \{ \sfrac 12 , \epsilon n^{-2} \}$.
	
	This completes the proof.
\end{proof}

Similarly, we upper bound the increase of the multilinear extension $L$ of modular function $\ell$ in each iteration.

\begin{lemma} \label{L bound}
	For every time $t \in \Tcal$, it holds that
	\begin{equation*}
		L( \yvec (t+\delta) ) - L( \yvec (t) ) \leq \delta \innprod{\lvec}{\zvec(t)}.
	\end{equation*}
\end{lemma}

\begin{proof}
	For every time $t \in \Tcal$, we have
	\begin{align*}
		L( \yvec (t+\delta) ) - L( \yvec (t) ) &= \innprod{\lvec}{ \yvec (t+\delta) - \yvec(t) } \\
		&= \delta \innprod{\lvec}{\zvec(t)} - \delta \innprod{\lvec}{\zvec(t) \circ \yvec(t) } \\
		&\leq \delta \innprod{\lvec}{\zvec(t)},
	\end{align*}
	where the first equality follows by the multilinear extension of a normalized modular function, the second equality follows by $\yvec(t+\delta) = \yvec(t) + \delta \cdot \zvec(t) \circ ( \ind_{\Omega} - \yvec(t) )$, and the inequality follows by the non-negativity of $\delta$, $\lvec$, $\zvec(t)$ and $\yvec(t)$.
	
	The proof is completed.
\end{proof}

We define a distorted objective function as $\Phi(t) := (1 - \delta)^{ \frac{1-t}{\delta} } G( \yvec(t) ) - L( \yvec(t) )$. Notice that $\Phi(t)$ varies the relative importance between $g$ and $\ell$ as the algorithm proceeds. Next, we lower bound the increase of the distorted objective function $\Phi(t)$ in each iteration.

\begin{lemma} \label{phi bound}
	If the event $\Ecal$ happens, then, for every time $t \in \Tcal$, it holds that
	\begin{equation*}
		\Phi(t + \delta) - \Phi(t) \geq \delta \left[ e^{-1} g(OPT) - \ell(OPT) \right] - 5\epsilon \delta M,
	\end{equation*}
	where $M = \max \left\{ \max_{e\in \Omega} g(e|\varnothing), -\min_{e\in \Omega} g(e|\Omega - e) \right\} > 0$ and $OPT \in \argmax \{ g(A) - \ell(A) : A \subseteq 2^{\Omega}, \ind_A \in \Pcal \}$.
\end{lemma}

\begin{proof}
	For every time $t \in \Tcal$, we have
	\begin{align*}
		\Phi&(t + \delta) - \Phi(t) \\ 
		&= (1 - \delta)^{ \frac{1-t}{\delta} - 1 } \left[ G( \yvec(t+\delta) ) - G( \yvec(t) ) \right] - \left[ L( \yvec(t+\delta) ) - L( \yvec(t) )  \right] + \delta (1 - \delta)^{ \frac{1-t}{\delta} - 1 } G( \yvec(t) ) \\
		&\geq \delta \innprod{ \zvec(t) }{ (1 - \delta)^{ \frac{1-t}{\delta} - 1 } \wvec(t) - \lvec } + \delta (1 - \delta)^{ \frac{1-t}{\delta} - 1 } G( \yvec(t) ) - 3 \epsilon \delta (1 - \delta)^{ \frac{1-t}{\delta} - 1 } M \\
		&\geq \delta \innprod{ \zvec(t) }{ (1 - \delta)^{ \frac{1-t}{\delta} - 1 } \wvec(t) - \lvec } + \delta (1 - \delta)^{ \frac{1-t}{\delta} - 1 } G( \yvec(t) ) - 3 \epsilon \delta M, 		
	\end{align*}
	where the equality follows by the definition of $\Phi(t)$, the first inequality follows by Lemma~\ref{G bound} and Lemma~\ref{L bound}, and the second inequality follows by $0 \leq (1 - \delta)^{ \frac{1-t}{\delta} - 1 } \leq 1$.
	
	Next, we lower bound the term $\innprod{ \zvec(t) }{ (1 - \delta)^{ \frac{1-t}{\delta} - 1 } \wvec(t) - \lvec }$. According to Algorithm~\ref{mcg distorted}, it holds that
	\begin{align*}
		\innprod{ \zvec(t) }{ (1 - \delta)^{ \frac{1-t}{\delta} - 1 } \wvec(t) - \lvec } &\geq \innprod{ \ind_{OPT} }{ (1 - \delta)^{ \frac{1-t}{\delta} - 1 } \wvec(t) - \lvec } \\
		&= (1 - \delta)^{ \frac{1-t}{\delta} - 1 } \innprod{\ind_{OPT}}{\wvec(t)} - \innprod{\ind_{OPT}}{\lvec} \\
		&= (1 - \delta)^{ \frac{1-t}{\delta} - 1 } \sum_{e \in OPT} w_e(t) - \ell (OPT).
	\end{align*}

	Meanwhile, when the event $\Ecal$ happens, $w_e(t) \geq \Ebb[ g(e| R_{\yvec(t)} ) ] - \frac{2\epsilon M}{n}$ holds for every element $e \in \Omega$ and every time $t \in \Tcal$. Thus, we have
	\begin{align*}
		\sum_{e \in OPT} w_e(t) &\geq \sum_{e \in OPT} \left[ \Ebb[ g(e| R_{\yvec(t)} ) ] - \frac{2\epsilon M}{n} \right] \\
		&\geq \Ebb \left[ g( R_{\yvec(t)} \cup OPT ) - g( R_{\yvec(t)} ) \right] - 2\epsilon M \\
		&= G( \yvec(t) \vee \ind_{OPT} ) - G( \yvec(t) ) - 2\epsilon M,
	\end{align*}
	where the second inequality follows by the submodularity of function $g$ and $|OPT| \leq n$, and the equality follows by the definition of multilinear extension.
	
	By Lemma~\ref{lb multi}, we can get
	\begin{align*}
		G( \yvec(t) \vee \ind_{OPT} ) &\geq \hat{g} ( \yvec(t) \vee \ind_{OPT} ) \\
		&= \int_{0}^{1} g( T_{\lambda} ( \yvec(t) \vee \ind_{OPT} ) ) \diff \lambda \\
		&\geq \int_{1 - (1-\delta)^{t/\delta}}^{1} g( T_{\lambda} ( \yvec(t) \vee \ind_{OPT} ) ) \diff \lambda \\
		&= (1-\delta)^{ \frac{t}{\delta} } g(OPT),
	\end{align*}
	where the first equality follows by the definition of the Lov\'{a}sz extension, the second equality holds since $T_{\lambda} ( \yvec(t) \vee \ind_{OPT} ) = OPT$ for any $\lambda \in [1 - (1-\delta)^{t/\delta}, 1]$, and the inequality follows by the non-negativity of function $g$.
	
	It follows that
	\begin{align*}
		\langle \zvec(t) &, (1 - \delta)^{ \frac{1-t}{\delta} - 1 } \wvec(t) - \lvec \rangle \\
		&\geq (1 - \delta)^{ \frac{1-t}{\delta} - 1 } \left[ G( \yvec(t) \vee \ind_{OPT} ) - G( \yvec(t) ) - 2\epsilon M \right] - \ell (OPT) \\
		&\geq (1 - \delta)^{ \frac{1}{\delta} - 1 } g(OPT) - \ell (OPT) - (1 - \delta)^{ \frac{1-t}{\delta} - 1 } G( \yvec(t) ) - 2\epsilon (1 - \delta)^{ \frac{1-t}{\delta} - 1 } M \\
		&\geq (1 - \delta)^{ \frac{1}{\delta} - 1 } g(OPT) - \ell (OPT) - (1 - \delta)^{ \frac{1-t}{\delta} - 1 } G( \yvec(t) ) - 2\epsilon M.
	\end{align*}
	Thus, we have
	\begin{align*}
		\Phi(t + \delta) - \Phi(t) &\geq \delta \left[ (1 - \delta)^{ \frac{1}{\delta} - 1 } g(OPT) - \ell (OPT) \right] - 5 \epsilon \delta M \\
		&\geq \delta \left[ e^{-1} g(OPT) - \ell (OPT) \right] - 5 \epsilon \delta M
	\end{align*}
	where the last inequality holds due to $\ln(1 - \delta) \geq - \frac{\delta}{1 - \delta}$ for any $\delta \in (0,1)$.
	
	The proof is completed.
\end{proof}

Now, we can analyze the performance guarantee of Algorithm~\ref{mcg distorted}.

\begin{lemma} \label{quality y(1)}
	If the event $\Ecal$ happens, then it holds that
	\begin{equation*}
		G( \yvec(1) ) - L( \yvec(1) ) \geq e^{-1} g(OPT) - \ell (OPT) - 5 \epsilon M,
	\end{equation*}
	where $M = \max \left\{ \max_{e\in \Omega} g(e|\varnothing), -\min_{e\in \Omega} g(e|\Omega - e) \right\} > 0$ and $OPT \in \argmax \{ g(A) - \ell(A) : A \subseteq 2^{\Omega}, \ind_A \in \Pcal \}$.
\end{lemma}

\begin{proof}
	By the definition of function $\Phi(t)$, we have $\Phi(1) = G( \yvec(1) ) - L( \yvec(1) )$ and $\Phi(0) = (1 - \delta)^{\frac 1\delta} g(\varnothing) - \ell(\varnothing) = (1 - \delta)^{\frac 1\delta} g(\varnothing) \geq 0$. It follows that
	\begin{align*}
		G( \yvec(1) ) - L( \yvec(1) ) &= \Phi(1) \\
		&= \Phi(0) + \sum_{i=0}^{\frac 1\delta - 1} \left[ \Phi( i\delta + \delta ) - \Phi( i\delta ) \right] \\
		&\geq \Phi(0) + \sum_{i=0}^{\frac 1\delta - 1} \left[ \delta e^{-1} g(OPT) - \delta \ell (OPT) - 5 \epsilon \delta M \right] \\
		&\geq e^{-1} g(OPT) - \ell (OPT) - 5 \epsilon M,
	\end{align*}
	where the first inequality follows by Lemma~\ref{phi bound}, and the second inequality holds due to $\Phi(0) \geq 0$.
\end{proof}

In each iteration, Algorithm~\ref{mcg distorted} requires $2nr$ value oracle queries. Thus, during the $\frac 1\delta$ iterations, Algorithm~\ref{mcg distorted} performs $2nr \delta^{-1} = O ( \frac{n^5}{\epsilon^3} \ln \frac{n^3}{\epsilon^2} )$ value oracle queries in total. Based on Lemma~\ref{high P event} and Lemma~\ref{quality y(1)}, we get the following theorem.

\begin{theorem}
	When Algorithm~\ref{mcg distorted} terminates, it produces a vector $\xvec \in \Pcal$ such that with high probability 
	\begin{equation*}
		G( \xvec ) - L( \xvec ) \geq e^{-1} g(OPT) - \ell (OPT) - 5 \epsilon M,
	\end{equation*}
	where $M = \max \left\{ \max_{e\in \Omega} g(e|\varnothing), -\min_{e\in \Omega} g(e|\Omega - e) \right\} > 0$ and $OPT \in \argmax \{ g(A) - \ell(A) : A \subseteq 2^{\Omega}, \ind_A \in \Pcal \}$. And, during the iteration, the total value oracle queries are $O ( \frac{n^5}{\epsilon^3} \ln \frac{n^3}{\epsilon^2} )$.
\end{theorem}

\subsection{Rounding}
When we restrict the polytope constraint to be a matroid polytope, there exists rounding techniques such as pipage rounding (\cite{vondrak2013symmetry}), which can produce an integral solution without lossing anything in the objective. Specifically, pipage rounding can output a random independent set $S \in \Ical$ of matroid $\Mcal = (\Omega, \Ical)$ such that
\begin{equation*}
	\Ebb [ g(S) - \ell (S) ] \geq G( \yvec(1) ) - L( \yvec(1) ) \geq e^{-1} g(OPT) - \ell (OPT) - 5 \epsilon M,
\end{equation*}
where $M = \max \left\{ \max_{e\in \Omega} g(e|\varnothing), -\min_{e\in \Omega} g(e|\Omega - e) \right\} > 0$ and $OPT \in \argmax \{ g(A) - \ell(A) : A \in \Ical \}$.

\section{Cardinality Constraint} \label{sec: cardinality constraint}
In this section, we present two algorithms for solving the regularized non-monotone maximization problem under a cadinality constraint, i.e.,
\begin{equation*} 
	\begin{alignedat}{1}
		\max_{A \subseteq \Omega} \quad & g(A) - \ell(A) \\ 
		\st \quad & |A| \leq k ,
	\end{alignedat}
\end{equation*}
where $\gfun$ is a non-monotone submodular function, $\ell \colon 2^\Omega \to \Rbb_+$ is a normalized modular function, and $k$ is the cardinality constraint. When $k = 1$, we can solve this problem by simply examining all elements in the ground set $\Omega$. Thus, we assume that $k \geq 2$. 

The key technique we use is designing a so-called \emph{distorted objective function}, which was introduced by \cite{harshaw2019submodular}. We introduce two important functions $\Phi$ and $\Psi$, which is crucial in our algorithms. Let $k$ be the cardinality constraint, for any $i = 0,1,\ldots,k$ and any set $T \subseteq \Omega$, we define
\begin{equation*}
	\Phi_i(T) := \ai g(T) - \ell(T).
\end{equation*}
Additionally, for any iteration $i = 0,1,\ldots,k-1$ of our algorithms, any set $T \subseteq \Omega$, and any element $e \in \Omega$, we define
\begin{equation*}
	\Psi_i(T,e) := \max \left\{ 0, \bi g(e|T) - \ell_e \right\}.
\end{equation*}

\subsection{\drg}
In this subsection, we present an algorithm which is based on the random greedy algorithm (\cite{buchbinder2014submodular}) and the distorted objective proposed by \cite{harshaw2019submodular}. Thus, we call this algorithm \drg.

\begin{algorithm}[htbp]
	\caption{\drg}
	\label{drg}
	\SetAlgoLined
	\KwIn{function $g$ and $\ell$, cardinality $k$}
	\KwOut{approximate solution $S_k$}
	Initialize $S_0 \leftarrow \varnothing$\;
	\For{$i = 0$ \KwTo $k-1$}{
		$M_i \leftarrow \argmax \left\{ \sum_{e \in B} \left[ \bi g(e|S_i) - \ell_e \right] : B \subseteq \Omega, |B| \leq k \right\}$ \;
		\textbf{with} probability $( 1 - |M_i|/k )$ \textbf{do} $S_{i+1} \leftarrow S_i$ \;
		\textbf{otherwise}	Let $e_i$ be a uniformly random element of $M_i$, and set $S_{i+1} \leftarrow S_i + e_i$ \;
	}
	\Return $S_k$
\end{algorithm}

Firstly, we consider the increase of the distorted
objective $\Phi_i$ in each iteration.

\begin{lemma} \label{increase Phi_i}
	In each iteration ($i=0,1,\ldots,k-1$) of Algorithm~\ref{drg}, it holds that
	\begin{equation*}
		\Ebb_{S_{i+1}} [ \Phi_{i+1}(S_{i+1}) ] -  \Ebb_{S_i} [ \Phi_i(S_i) ] \geq \frac 1k \bi \Ebb_{S_i} [ g(OPT \cup S_i) ] - \frac 1k \ell(OPT),
	\end{equation*}
	where $OPT \in \argmax \{ g(A) - \ell(A) : A \subseteq \Omega, |A| \leq k \}$.
\end{lemma}

\begin{proof}
	Since $S_{i+1} = S_i$ or $S_{i+1} = S_i + e_i$ when $S_i$ is given, we have
	\begin{align*}
		\Ebb&_{S_{i+1}} \left[ \Phi_{i+1}(S_{i+1}) - \Phi_i(S_i) | S_i \right] \\
		&= \pr [ S_{i+1} = S_i | S_i ] \cdot [ \Phi_{i+1}(S_i) - \Phi_i(S_i) ] + \sum_{e \in M_i} \pr [ S_{i+1} = S_i + e | S_i ] \cdot [ \Phi_{i+1}(S_i + e) - \Phi_i(S_i) ] .
	\end{align*}

	According to the settings of Algorithm~\ref{drg}, we get that $\pr [ S_{i+1} = S_i | S_i ] = 1 - \frac{|M_i|}{k} $ and $\pr [ S_{i+1} = S_i + e | S_i ] = \frac 1k$ for any $e \in M_i$. Meanwhile, by the definition of function $\Phi_i$, it holds that
	\begin{equation*}
		\Phi_{i+1}(S_i) - \Phi_i(S_i) = \frac 1k \bi g(S_i),
	\end{equation*}
	and
	\begin{equation*}
		\Phi_{i+1}(S_i + e) - \Phi_i(S_i) = \bi g(e|S_i) - \ell_e + \frac 1k \bi g(S_i).
	\end{equation*}

	Thus, we have
	\begin{align*}
		\Ebb&_{S_{i+1}} \left[ \Phi_{i+1}(S_{i+1}) - \Phi_i(S_i) | S_i \right] \\ 
		&= \left( 1 - \frac{|M_i|}{k} \right) \cdot \frac 1k \bi g(S_i) \\
		&\quad + \frac 1k \sum_{e \in M_i} \left[ \bi g(e|S_i) - \ell_e + \frac 1k \bi g(S_i) \right] \\
		&= \frac 1k \sum_{e \in M_i} \left[ \bi g(e|S_i) - \ell_e \right] + \frac 1k \bi g(S_i) \\
		&\geq \frac 1k \sum_{e \in OPT} \left[ \bi g(e|S_i) - \ell_e \right] + \frac 1k \bi g(S_i) \\
		&\geq \frac 1k \bi [ g(OPT \cup S_i) - g(S_i) ] - \frac 1k \ell(OPT) + \frac 1k \bi g(S_i) \\
		&= \frac 1k \bi g(OPT \cup S_i) - \frac 1k \ell(OPT),
	\end{align*}
	where the first inequality follows by the choice of $M_i$ and $|OPT| \leq k$, and the second inequality follows by the submodularity of function $g$.
	
	It follows that
	\begin{align*}
		\Ebb_{S_{i+1}} [ \Phi_{i+1}(S_{i+1}) ] -  \Ebb_{S_i} [ \Phi_i(S_i) ] &= \Ebb_{S_i} \left[ \Ebb_{S_{i+1}} [ \Phi_{i+1}(S_{i+1}) - \Phi_i(S_i) | S_i ] \right] \\
		&\geq \frac 1k \bi \Ebb_{S_i} [ g(OPT \cup S_i) ] - \frac 1k \ell(OPT),
	\end{align*}
	which completes the proof.
\end{proof}

Next, we bound the term $\Ebb_{S_i} [ g(OPT \cup S_i) ]$ from below with respect to $g(OPT)$. In order to achieve this, we need the following lemma.

\begin{lemma}[Lemma~2.2 in \cite{buchbinder2014submodular}] \label{buch lemma}
	Let $f \colon 2^\Omega \to \Rbb$ be a submodular function. Denote by $A(p)$ a random subset of $A \subseteq \Omega$ where each element appears with probability at most $p$ (not necessarily independently). Then, $\Ebb_{A(p)} [ f( A(p) ) ] \geq (1 - p) f(\varnothing)$.
\end{lemma}

Now, we can lower bound $\Ebb_{S_i} [ g(OPT \cup S_i) ]$ in terms of $g(OPT)$.

\begin{lemma} \label{lb g(Si)} 
	For every $i = 0,1,\ldots,k-1$, it holds that $\Ebb_{S_i} [ g(OPT \cup S_i) ] \geq \left( 1 -\frac 1k \right)^i g(OPT)$.
\end{lemma}

\begin{proof}
	Suppose that $e$ is an arbitrary element in ground set $\Omega$. When the event $e \notin S_i$ happens, then $e_i \neq e$ implies $e \notin S_{i+1}$. Thus, we have $\pr [ e \notin S_{i+1} | e \notin S_i ] \geq \pr [ e_i \neq e | e \notin S_i ]$.
	
	If $e \notin M_i$, then it holds that $\pr [ e_i \neq e | e \notin S_i ] = 1$. If $e \in M_i$, then it holds that $\pr [ e_i \neq e | e \notin S_i ] = 1 - \frac 1k$. As a result, for any $e \in \Omega$ and any $i = 0,1,\ldots,k-1$, it holds that $\pr [ e \notin S_{i+1} | e \notin S_i ] \geq 1 - \frac 1k$.
	
	It follows that, for every $i = 1,2,\ldots,k-1$,
	\begin{align*}
		\pr [ e \notin S_i ] &= \pr [ e \notin S_0,e \notin S_1,e \notin S_2,\ldots,e \notin S_{i-1},e \notin S_i ] \\
		&= \pr [ e \notin S_0 ] \cdot \pr [ e \notin S_1 | e \notin S_0 ] \cdot \pr [ e \notin S_2 | e \notin S_0,e \notin S_1 ] \cdots \pr [ e \notin S_i | e \notin S_0,e \notin S_1,\ldots,e \notin S_{i-1} ] \\
		&= \pr [ e \notin S_0 ] \cdot \pr [ e \notin S_1 | e \notin S_0 ] \cdot \pr [ e \notin S_2 | e \notin S_1 ] \cdots \pr [ e \notin S_i | e \notin S_{i-1} ] \\
		&= \pr [ e \notin S_1 | e \notin S_0 ] \cdot \pr [ e \notin S_2 | e \notin S_1 ] \cdots \pr [ e \notin S_i | e \notin S_{i-1} ] \\
		&\geq \left( 1 - \frac 1k \right)^i,
	\end{align*}
	where the third equality holds since $e \notin S_j$ implies that $e \notin S_r$ for every $r = 0,1,\ldots,j-1$, and the fourth equality follows by $S_0 = \varnothing$. Thus, we get that $\pr [ e \in S_i ] = 1 - \pr [ e \notin S_i ] \leq 1 - \left( 1 - \frac 1k \right)^i $ for every element $e \in \Omega$ and every $i = 1,2,\ldots,k-1$.
	
	Let $h \colon 2^\Omega \to \Rbb_+$ be the function $h(A) := g(OPT \cup A)$ for every $A \subseteq \Omega$. Note that $h$ is still a submodular function. Thus, by Lemma~\ref{buch lemma}, we have $\Ebb_{S_i} [ h(S_i) ] \geq \left( 1 - \frac 1k \right)^i h(\varnothing)$, which indicates that $\Ebb_{S_i} [ g(OPT \cup S_i) ] \geq \left( 1 -\frac 1k \right)^i g(OPT)$ for every $i = 1,2,\ldots,k-1$. Since $\Ebb_{S_0} [ g(OPT \cup S_0) ] = g(OPT)$, the proof is completed.
\end{proof}

With the help of abovementioned lemmas, we can analyze the performance guarantee of Algorihtm~\ref{drg}.

\begin{theorem} \label{drg guarantee}
	When Algorithm~\ref{drg} terminates, it returns a feasible set $S_k$ with
	\begin{equation*}
		\Ebb_{S_k} [ g(S_k) - \ell(S_k) ] \geq \left( 1 -\frac 1k \right)^{k-1} g(OPT) - \ell(OPT) \geq e^{-1} g(OPT) - \ell(OPT),
	\end{equation*}
	where $OPT \in \argmax \{ g(A) - \ell(A) : A \subseteq \Omega, |A| \leq k \}$. And, during the iteration, the total value oracle queries are $O(kn)$.
\end{theorem}

\begin{proof}
	According to Lemma~\ref{increase Phi_i} and Lemma~\ref{lb g(Si)}, we have
	\begin{equation*}
		\Ebb_{S_{i+1}} [ \Phi_{i+1}(S_{i+1}) ] -  \Ebb_{S_i} [ \Phi_i(S_i) ] \geq \frac 1k \left( 1 -\frac 1k \right)^{k-1} g(OPT) - \frac 1k \ell(OPT),
	\end{equation*}
	for every $i = 0,1,\ldots,k-1$.
	
	By the definition of function $\Phi_i$, we get
	\begin{equation*}
		\Phi_k(S_k) = g(S_k) - \ell(S_k),
	\end{equation*}
 	and
 	\begin{equation*}
 		\Phi_0(S_0) = \left( 1 -\frac 1k \right)^k g(\varnothing) - \ell(\varnothing) = \left( 1 -\frac 1k \right)^k g(\varnothing) \geq 0.
 	\end{equation*}
 	
 	It follows that
 	\begin{align*}
 		\Ebb_{S_k} [ g(S_k) - \ell(S_k) ] &= \Ebb_{S_k} [ \Phi_k(S_k) ] = \Ebb_{S_0} [ \Phi_k(S_0) ] + \sum_{i=0}^{k-1} \left\{ \Ebb_{S_{i+1}} [ \Phi_{i+1}(S_{i+1}) ] -  \Ebb_{S_i} [ \Phi_i(S_i) ] \right\} \\
 		&\geq \sum_{i=0}^{k-1} \left[ \frac 1k \left( 1 -\frac 1k \right)^{k-1} g(OPT) - \frac 1k \ell(OPT) \right] \\
 		&= \left( 1 -\frac 1k \right)^{k-1} g(OPT) - \ell(OPT) \\
 		&\geq e^{-1} g(OPT) - \ell(OPT),
 	\end{align*}
 	where the last inequality holds due to $\ln ( 1 - \frac 1k ) \geq - \frac{1}{k-1}$ for every $k \geq 2$.
 	
 	As for the number of value oracle queries, we notice that $O(n)$ value oracle queries are needed in each iteration. Thus, during $k$ iterations, Algorithm~\ref{drg} requires $O(kn)$ value oracle queries in total. 
\end{proof}

\subsection{\drsg}
In this subsection, we will prove the following theorem.

\begin{theorem}
	There exists a randomized algorithm that given a non-monotone submodular funtion $\gfun$, a normalized modular function $\ell \colon 2^\Omega \to \Rbb_+$, and parameters $k \geq 2$ and $\epsilon \in (0, e^{-1})$, returns a feasible solution $S \subseteq \Omega$ with
	\begin{equation*}
		\Ebb_S [g(S) - \ell(S)] \geq (e^{-1} - \epsilon) g(OPT) - \ell(OPT),
	\end{equation*}
	where $OPT \in \argmax \{ g(A) - \ell(A) : A \subseteq \Omega, |A| \leq k \}$. And, the algorithm performs $O( \frac{n}{\epsilon^2} \ln \frac 1\epsilon )$ value oracle queries.
\end{theorem}

Since function $\alpha(x) := \frac{8}{x^2} \ln \frac 2x$ is strictly decreasing when $x \in (0,2)$, then there exists a unique $x \in (0,2)$ such that $\alpha(x) = k$. In this subsection, we denote by $\delta$ the unique solution of this equation, i.e, $\frac{8}{\delta^2} \ln \frac 2\delta = k$. When $\epsilon \in (0,\delta]$, the number of value oracle queries which Algorithm~\ref{drg} requires is $O(kn) = O(\frac{8n}{\epsilon^2} \ln \frac 2\epsilon) = O(\frac{n}{\epsilon^2} \ln \frac 1\epsilon)$ due to $k = \frac{8}{\delta^2} \ln \frac 2\delta \leq \frac{8}{\epsilon^2} \ln \frac 2\epsilon$. Thus, we only need to consider the case in which $\epsilon \in (\delta,e^{-1})$. 

We propose an algorithm to deal with this situation, i.e, $\epsilon \in (\delta,e^{-1})$. Our algorithm, which is presented as Algrithm~\ref{drsg}, is based on the Random Sampling algorithm (\cite{buchbinder2017comparing}) and the distorted objective. Notice that when $\epsilon \in (\delta,e^{-1})$, it holds that $p \in (0,1]$ and $1 \leq s \leq \lceil pn \rceil$.

\begin{algorithm}[htbp]
	\caption{\drsg}
	\label{drsg}
	\SetAlgoLined
	\KwIn{function $g$ and $\ell$, cardinality $k$, error $\epsilon \in (\delta, e^{-1})$}
	\KwOut{approximate solution $S_k$}
	Initialize $S_0 \leftarrow \varnothing$, $p \leftarrow \frac{8}{k \epsilon^2} \ln \frac{2}{\epsilon}$, $s \leftarrow \frac kn \lceil pn \rceil$\;
	\For{$i = 0$ \KwTo $k-1$}{
		Let $M_i$ be a uniformly random set containing $\lceil pn \rceil$ elements of $\Omega$ \; 
		Let $d_i$ be a uniformly random value from the range $(0,s]$ \;  
		Let $e_i$ be the element of $M_i$ with the $\lceil d_i \rceil$-th largest distorted marginal contribution to $S_i$ \; 
		\tcp{The distorted marginal gain with respect to $S_i$ is equal to $\bi g(e|S_i) - \ell_{e}$ for every element $e \in \Omega$.}
		\eIf{ $\bi g(e_i|S_i) - \ell_{e_i} > 0 $ }{
			$S_{i+1} \leftarrow S_i + e_i$ \;
		}{
			$S_{i+1} \leftarrow S_i$ \;
		}
	}
	\Return $S_k$ \;
\end{algorithm}

Suppose that $S_i$ is given, we sort all the elements of $\Omega$ in order of non-increasing distorted marginal gain. We assume that 
\begin{equation*}
	\bi g(v_1|S_i) - \ell_{v_1} \geq \bi g(v_2|S_i) - \ell_{v_2} \geq \cdots \geq \bi g(v_n|S_i) - \ell_{v_n},
\end{equation*}
where $\Omega = \{ v_1,v_2,\ldots,v_n \}$. Moreover, we define $n$ random variables $X_j (j = 1,2,\ldots,n)$ as
\begin{equation*}
	X_j = \left\{ 
	\begin{aligned}
		1, \quad & e_i = v_j; \\
		0, \quad & \text{otherwise}.
	\end{aligned}
	\right.
\end{equation*}
Then, we can get the following two lemmas for the same reason as is shown in the proof of Lemma~4.3 and Lemma~4.4 in \cite{buchbinder2017comparing}.

\begin{lemma}[Lemma~4.3 in \cite{buchbinder2017comparing}] \label{sum of indicator}
	For every $i = 0,1,\ldots,k-1$, it holds that $\Ebb_{e_i} [ \sum_{j=1}^{k} X_j | S_i ] \geq 1 - \epsilon$.
\end{lemma}

\begin{lemma}[Lemma~4.4 in \cite{buchbinder2017comparing}]
	For every $i = 0,1,\ldots,k-1$, it holds that $\Ebb_{e_i} [ X_j | S_i ]$ is a non-increasing function of $j$.
\end{lemma}

Next, we consider the increase of the distorted objective $\Phi_i$ in each iteration.

\begin{lemma} \label{marginal gain phi_k}
	In each iteration ($i=0,1,\ldots,k-1$) of Algorithm~\ref{drsg}, it holds that
	\begin{equation*}
		\Phi_{i+1}(S_{i+1}) - \Phi_i(S_i) = \Psi_i(S_i, e_i) + \frac 1k \bi g(S_i).
	\end{equation*}
\end{lemma}

\begin{proof}
	We consider two cases.
	\begin{enumerate}[(i)]
		\item If $S_{i+1} = S_i$, then $\Psi_i(S_i, e_i) = 0$. By the definition of $\Phi_i$, we have
		\begin{align*}
			\Phi_{i+1}(S_{i+1}) - \Phi_i(S_i) &= \left[ \bi g(S_i) - \ell(S_i) \right] - \left[ \ai g(S_i) - \ell(S_i) \right] \\
			&= \frac 1k \bi g(S_i) \\
			&= \Psi_i(S_i, e_i) + \frac 1k \bi g(S_i).
		\end{align*}
		
		\item If $S_{i+1} = S_i + e_i$, then $\Psi_i(S_i, e_i) = \bi g(e_i|S_i) - \ell_{e_i} > 0$. It follows that
		\begin{align*}
			\Phi_{i+1}&(S_{i+1}) - \Phi_i(S_i) \\
			&= \left[ \bi g(S_i + e_i) - \ell(S_i + e_i) \right] - \left[ \ai g(S_i) - \ell(S_i) \right] \\
			&= \bi g(e_i|S_i) - \ell_{e_i} + \frac 1k \bi g(S_i) \\
			&= \Psi_i(S_i, e_i) + \frac 1k \bi g(S_i).
		\end{align*}
	\end{enumerate}
	The proof is completed.
\end{proof}

We then lower bound the expected increase of the distorted objective $\Phi_i$ in each iteration.

\begin{lemma} \label{expected increase phi}
	For every $i = 0,1,\cdots,k-1$, it holds that
	\begin{align*}
		\Ebb_{S_{i+1}} &[ \Phi_{i+1} (S_{i+1}) ] - \Ebb_{S_i} [ \Phi_i (S_i) ] \geq \frac{1-\epsilon}{k} \left\{ \bi \Ebb_{S_i} [ g (OPT \cup S_i) ] - \ell(OPT) \right\},
	\end{align*}
	where $OPT \in \argmax \{ g(A) - \ell(A) : A \subseteq \Omega, |A| \leq k \}$.
\end{lemma}

\begin{proof}
	By Lemma~\ref{marginal gain phi_k}, it holds that
	\begin{equation*}
		\Ebb_{e_i} \left[ \Phi_{i+1} (S_{i+1}) - \Phi_i (S_i) | S_i \right] = \Ebb_{e_i} [ \Psi_i(S_i,e_i) | S_i ] + \frac 1k \bi g(S_i).
	\end{equation*}

	Based on the definition of $X_j$, we immedietly get $\Psi_i(S_i, e_i) = \sum_{j=1}^n X_j \cdot \Psi_i(S_i, v_j)$. Since $\Psi_i(S_i, v_j)$ is non-negative for any $j = 1,2,\ldots,n$, it follows that $\Psi_i(S_i, e_i) \geq \sum_{j=1}^k X_j \cdot \Psi_i(S_i, v_j)$. By the linearity of expectation, we have \[ \Ebb_{e_i} [\Psi_i(S_i, e_i) | S_i] \geq \sum_{j=1}^k \Ebb_{e_i} [X_j | S_i] \cdot \Psi_i(S_i, v_j). \]
	
	Since $\Ebb_{e_i} [X_j | S_i]$ and $\Psi_i(S_i, v_j)$ are both non-increasing functions of $j$, we have
	
	\begin{align*}
		\Ebb_{e_i} [\Psi_i(S_i, e_i) | S_i] &\geq \sum_{j=1}^k \Ebb_{e_i} [X_j | S_i] \cdot \frac 1k \sum_{j=1}^{k}\Psi_i(S_i, v_j) \\
		&\geq \frac{1-\epsilon}{k} \sum_{e \in OPT} \Psi_i(S_i, e) \\
		&= \frac{1-\epsilon}{k} \sum_{e \in OPT} \max \left\{ 0, \bi g(e|S_i) - \ell_{e} \right\} \\
		&\geq \frac{1-\epsilon}{k} \sum_{e \in OPT}  \left[ \bi g(e|S_i) - \ell_{e} \right] \\
		&\geq \frac{1-\epsilon}{k} \left\{ \bi [ g(OPT \cup S_i) - g(S_i) ] - \ell(OPT) \right\} ,
	\end{align*}
	where the first inequality follows by Chebyshev's sum inequality\footnote{If $a_1 \geq a_2 \geq \cdots \geq a_n$ and $b_1 \geq b_2 \geq \cdots \geq b_n$, then it holds that $\frac 1n \sum_{k=1}^{n} a_k \cdot b_k \geq \left( \frac 1n \sum_{k=1}^{n} a_k \right) \left( \frac 1n \sum_{k=1}^{n} b_k \right)$.}, the second inequality follows by Lemma~\ref{sum of indicator}, and the last inequality follows by the submodularity of function $g$.
	
	Thus, we obtain
	\begin{align*}
		\Ebb_{e_i} &\left[ \Phi_{i+1} (S_{i+1}) - \Phi_i (S_i) | S_i \right] \\ 
		&\geq \frac{1-\epsilon}{k} \left\{ \bi [ g(OPT \cup S_i) - g(S_i) ] - \ell(OPT) \right\} + \frac 1k \bi g(S_i) \\
		&= \frac{1-\epsilon}{k} \left[ \bi g(OPT \cup S_i) - \ell(OPT) \right] + \frac{\epsilon}{k}  \bi g(S_i) \\
		&\geq \frac{1-\epsilon}{k} \left[ \bi g(OPT \cup S_i) - \ell(OPT) \right],
	\end{align*}
	where the last inequality follows by the non-negativity of function $g$.
	
	It follows that
	\begin{align*}
		\Ebb_{S_{i+1}} &[ \Phi_{i+1} (S_{i+1}) ] - \Ebb_{S_i} [ \Phi_i (S_i) ] \\ 
		&= \Ebb_{S_i} \left[ \Ebb_{e_i} \left[ \Phi_{i+1} (S_{i+1}) - \Phi_i (S_i) | S_i \right] \right] \\
		&\geq \frac{1-\epsilon}{k} \left\{ \bi \Ebb_{S_i} [ g(OPT \cup S_i) ] - \ell(OPT) \right\},
	\end{align*}
	which conludes the proof.
\end{proof}

Again, we need to bound the term $\Ebb_{S_i} [ g(OPT \cup S_i) ]$ from below with respect to $g(OPT)$.

\begin{lemma} \label{lb g(Si) sample} 
	For every $i = 0,1,\ldots,k-1$, it holds that $\Ebb_{S_i} [ g(OPT \cup S_i) ] \geq \left( 1 -\frac 1k \right)^i g(OPT)$.
\end{lemma}

\begin{proof}
	Suppose that $e$ is an arbitrary element in ground set $\Omega$. When the event $e \notin S_i$ happens, then $e_i \neq e$ implies $e \notin S_{i+1}$. Thus, we have 
	\begin{align*}
		\pr [ e \notin S_{i+1} | e \notin S_i ] &\geq \pr [ e_i \neq e | e \notin S_i ] \\
		&= \pr[ e \notin M_i ] \cdot \pr[ e_i \neq e | e \notin S_i, e \notin M_i ] + \pr[ e \in M_i ] \cdot \pr[ e_i \neq e | e \notin S_i, e \in M_i ].
	\end{align*}

	According to the settings of Algorithm~\ref{drsg}, it holds that
	\begin{equation*}
		\pr [ e \in M_i ] = \left. \binom{n-1}{\lceil pn \rceil - 1} \middle/ \binom{n}{\lceil pn \rceil} \right. = \frac{\lceil pn \rceil}{n},
	\end{equation*}
	and
	\begin{equation*}
		\pr[ e_i \neq e | e \notin S_i, e \notin M_i ] = 1.
	\end{equation*}

	Now, we consider the conditional probability $\pr[ e_i \neq e | e \notin S_i, e \in M_i ]$. If $s$ is an integer, we have  $\pr[ e_i \neq e | e \notin S_i, e \in M_i ] = \frac 1s$. If $s$ is not an integer, we have  $\pr[ e_i \neq e | e \notin S_i, e \in M_i ] = \frac 1s \text{ or } \frac{s - \lfloor s \rfloor}{s}$. Thus, it always holds that $\pr[ e_i \neq e | e \notin S_i, e \in M_i ] \leq \frac 1s$.
	
	We obtain
	\begin{align*}
		\pr [ e \notin S_{i+1} | e \notin S_i ] &\geq \left( 1 - \frac{\lceil pn \rceil}{n} \right) + \frac{\lceil pn \rceil}{n} \cdot \left( 1 - \frac 1s \right) \\
		&= 1 - \frac{\lceil pn \rceil}{sn} \\
		&= 1 - \frac 1k,
	\end{align*}
	where the last equality holds due to $s = \frac kn \lceil pn \rceil$.
	
	It follows that, for every $i = 1,2,\ldots,k-1$,
	\begin{align*}
		\pr [ e \notin S_i ] &= \pr [ e \notin S_0,e \notin S_1,e \notin S_2,\ldots,e \notin S_{i-1},e \notin S_i ] \\
		&= \pr [ e \notin S_0 ] \cdot \pr [ e \notin S_1 | e \notin S_0 ] \cdot \pr [ e \notin S_2 | e \notin S_0,e \notin S_1 ] \cdots \pr [ e \notin S_i | e \notin S_0,e \notin S_1,\ldots,e \notin S_{i-1} ] \\
		&= \pr [ e \notin S_0 ] \cdot \pr [ e \notin S_1 | e \notin S_0 ] \cdot \pr [ e \notin S_2 | e \notin S_1 ] \cdots \pr [ e \notin S_i | e \notin S_{i-1} ] \\
		&= \pr [ e \notin S_1 | e \notin S_0 ] \cdot \pr [ e \notin S_2 | e \notin S_1 ] \cdots \pr [ e \notin S_i | e \notin S_{i-1} ] \\
		&\geq \left( 1 - \frac 1k \right)^i,
	\end{align*}
	where the third equality holds since $e \notin S_j$ implies that $e \notin S_r$ for every $r = 0,1,\ldots,j-1$, and the fourth equality follows by $S_0 = \varnothing$. Thus, we get that $\pr [ e \in S_i ] = 1 - \pr [ e \notin S_i ] \leq 1 - \left( 1 - \frac 1k \right)^i $ for every element $e \in \Omega$ and every $i = 1,2,\ldots,k-1$.
	
	Let $h \colon 2^\Omega \to \Rbb_+$ be the function $h(A) := g(OPT \cup A)$ for every $A \subseteq \Omega$. Note that $h$ is still a submodular function. Thus, by Lemma~\ref{buch lemma}, we have $\Ebb_{S_i} [ h(S_i) ] \geq \left( 1 - \frac 1k \right)^i h(\varnothing)$, which indicates that $\Ebb_{S_i} [ g(OPT \cup S_i) ] \geq \left( 1 -\frac 1k \right)^i f(OPT)$ for every $i = 1,2,\ldots,k-1$. Since $\Ebb_{S_0} [ g(OPT \cup S_0) ] = g(OPT)$, the proof is completed.
\end{proof}

Finally, we analyze the performance guarantee of Algorihtm~\ref{drsg}.

\begin{theorem} \label{drsg guarantee}
	When Algorithm~\ref{drsg} terminates, it returns a feasible set $S_k$ with
	\begin{equation*}
		\Ebb_{S_k} [ g(S_k) - \ell(S_k) ] \geq (1-\epsilon) \left[ \left( 1 -\frac 1k \right)^{k-1} g(OPT) - \ell(OPT) \right] \geq (e^{-1} - \epsilon) g(OPT) - \ell(OPT),
	\end{equation*}
	where $OPT \in \argmax \{ g(A) - \ell(A) : A \subseteq \Omega, |A| \leq k \}$. And, during the iteration, the total value oracle queries are $O(\frac{n}{\epsilon^2} \ln \frac 1\epsilon)$.
\end{theorem}

\begin{proof}
	According to Lemma~\ref{expected increase phi} and Lemma~\ref{lb g(Si) sample}, we have
	\begin{equation*}
		\Ebb_{S_{i+1}} [ \Phi_{i+1}(S_{i+1}) ] -  \Ebb_{S_i} [ \Phi_i(S_i) ] \geq \frac{1-\epsilon}{k} \left[ \left( 1 -\frac 1k \right)^{k-1} g(OPT) - \ell(OPT) \right] ,
	\end{equation*}
	for every $i = 0,1,\ldots,k-1$.
	
	Since $\Phi_k(S_k) = g(S_k) - \ell(S_k)$ and $\Phi_0(S_0) = \left( 1 -\frac 1k \right)^k g(\varnothing) \geq 0$, we have
	\begin{align*}
		\Ebb_{S_k} [ g(S_k) - \ell(S_k) ] &= \Ebb_{S_k} [ \Phi_k(S_k) ] \\ 
		&= \Ebb_{S_0} [ \Phi_k(S_0) ] + \sum_{i=0}^{k-1} \left\{ \Ebb_{S_{i+1}} [ \Phi_{i+1}(S_{i+1}) ] -  \Ebb_{S_i} [ \Phi_i(S_i) ] \right\} \\
		&\geq \frac{1-\epsilon}{k} \sum_{i=0}^{k-1} \left[ \left( 1 -\frac 1k \right)^{k-1} g(OPT) - \ell(OPT) \right] \\
		&= (1-\epsilon) \left[ \left( 1 -\frac 1k \right)^{k-1} g(OPT) - \ell(OPT) \right] \\
		&\geq (1-\epsilon) [ e^{-1} g(OPT) - \ell(OPT) ] \\
		&= (e^{-1} - e^{-1}\epsilon) g(OPT) - \ell(OPT) + \epsilon \ell(OPT) \\
		&\geq (e^{-1} - \epsilon) g(OPT) - \ell(OPT)
	\end{align*}
	where the penultimate inequality holds due to $\ln ( 1 - \frac 1k ) \geq - \frac{1}{k-1}$ for every $k \geq 2$, and the last inequality follows by $e^{-1} \leq 1$ and the non-negativity of function $\ell$.
	
	As for the number of value oracle queries, we notice that $O(\lceil pn \rceil)$ value oracle queries are needed in each iteration. Thus, during $k$ iterations, Algorithm~\ref{drsg} requires $O(k \lceil pn \rceil) = O(\frac{n}{\epsilon^2} \ln \frac 1\epsilon)$ value oracle queries in total. 
\end{proof}

\section{Unconstrained Problem} \label{sec: unconstrained}
In this section, we propose an algorithm for solving the regularized non-monotone maximization problem under no constraint, i.e.,
\begin{equation*} 
	\max_{A \subseteq \Omega} \quad g(A) - \ell(A),
\end{equation*}
where $\gfun$ is a non-monotone submodular function, $\ell \colon 2^\Omega \to \Rbb_+$ is a normalized modular function.

Our algorithm, which is presented as Algorithm~\ref{unconstrained drg}, is also based on the distorted objective. Again, we introduce two auxiliary functions $\Phi$ and $\Psi$. For any $i = 0,1,\ldots,n$ and any set $T \subseteq \Omega$, we define
\begin{equation*}
	\Phi_i(T) := \ain g(T) - \ell(T).
\end{equation*}
Additionally, for any iteration $i = 0,1,\ldots,n-1$ of our algorithms, any set $T \subseteq \Omega$, and any element $e \in \Omega$, we define
\begin{equation*}
	\Psi_i(T,e) := \max \left\{ 0, \bin g(e|T) - \ell_e \right\}.
\end{equation*}

\begin{algorithm}[htbp]
	\caption{\udrg}
	\label{unconstrained drg}
	\SetAlgoLined
	\KwIn{function $g$ and $\ell$}
	\KwOut{approximate solution $S_n$}
	Initialize $S_0 \leftarrow \varnothing$\;
	\For{$i = 0$ \KwTo $n-1$}{
		Let $e_i$ be a uniformly random element of $\Omega$ \;
		\eIf{ $\bin g(e_i|S_i) - \ell_{e_i} > 0$ }{
			$S_{i+1} \leftarrow S_i + e_i$ \;
		}{
			$S_{i+1} \leftarrow S_i$ \;
		}
	}
	\Return $S_n$ \;
\end{algorithm}

Though the unconstrained setting is a special case of the maximization problem under a cardinality constraint (i.e., $k=n$), Algorithm~\ref{unconstrained drg} is much simpler than Algorithm~\ref{drg} and Algorithm~3. Moreover, we will show that Algorithm~\ref{unconstrained drg} has the same performance guarantee as Algorithm~\ref{drg} and Algorithm~3 do.

Firstly, we consider the increase of the distorted objective $\Phi_i$ in each iteration. In the same way of proving Lemma~\ref{marginal gain phi_k}, we can obtain the following lemma.

\begin{lemma} \label{marginal gain phi_i}
	In each iteration ($i=0,1,\ldots,n-1$) of Algorithm~\ref{unconstrained drg}, it holds that
	\begin{equation*}
		\Phi_{i+1}(S_{i+1}) - \Phi_i(S_i) = \Psi_i(S_i, e_i) + \frac 1n \bin g(S_i).
	\end{equation*}
\end{lemma}

Next, we consider the lower bound of the term $\Psi_i(S_i, e_i)$.

\begin{lemma} \label{expected lb psi}
	In each iteration ($i=0,1,\ldots,n-1$) of Algorithm~\ref{unconstrained drg}, it holds that
	\begin{equation*}
		\Ebb_{(S_i, e_i)} [ \Psi_i(S_i, e_i) ] \geq \frac 1n \left( 1 - \frac 1n \right)^{n-1} g(OPT) - \frac 1n \ell(OPT) - \frac 1n \bin \Ebb_{S_i} [g(S_i)] ,
	\end{equation*}
	where $OPT \in \argmax \{ g(A) - \ell(A) : A \subseteq \Omega \}$
\end{lemma}

\begin{proof}
	We notice that
	\begin{align*}
		\Ebb_{e_i} [ \Psi_i(S_i, e_i) | S_i ] &= \sum_{e \in \Omega} \pr[ e_i = e | S_i ] \cdot \Psi_i(S_i, e) \\
		&= \frac 1n \sum_{e \in \Omega} \Psi_i(S_i, e) \\
		&\geq \frac 1n \sum_{e \in OPT} \Psi_i(S_i, e) \\
		&= \frac 1n \sum_{e \in OPT} \max \left\{ 0, \bin g(e|S_i) - \ell_{e} \right\} \\
		&\geq \frac 1n \sum_{e \in OPT} \left[ \bin g(e|S_i) - \ell_{e} \right] \\
		&\geq  \frac 1n \bin [ g(OPT \cup S_i) - g(S_i) ] - \frac 1n \ell(OPT) ,
	\end{align*}
	where the first inequality follows by the non-negativity of $\Psi_i(S_i, e)$ for every $e \in \Omega$, and the last inequality follows by the submodularity of function $g$.
	
	It follows that
	\begin{align*}
		\Ebb&_{(S_i, e_i)} [ \Psi_i(S_i, e_i) ] \\
		&= \Ebb_{S_i} \left[ \Ebb_{e_i} [ \Psi_i(S_i, e_i) | S_i ] \right] \\
		&\geq \frac 1n \bin \Ebb_{S_i} [ g(OPT \cup S_i) ] - \frac 1n \ell(OPT) - \frac 1n \bin \Ebb_{S_i} [ g(S_i) ].
	\end{align*}

	Following similar argument in the proof of Lemma~\ref{lb g(Si)}, we can obtain that $\Ebb_{S_i} [ g(OPT \cup S_i) ] \geq \left( 1 - \frac 1n \right)^i g(OPT)$ holds for every $i=0,1,\ldots,n-1$. Thus, we have 
	\begin{equation*}
		\Ebb_{(S_i, e_i)} [ \Psi_i(S_i, e_i) ] \geq \frac 1n \left( 1 - \frac 1n \right)^{n-1} g(OPT) - \frac 1n \ell(OPT) - \frac 1n \bin \Ebb_{S_i} [g(S_i)] ,
	\end{equation*}
	which concludes the proof.
\end{proof}

Following similar argument in the proof of Lemma~\ref{drg guarantee}, we can obtain the performance guarantee of Algorithm~\ref{unconstrained drg} based on Lemma~\ref{marginal gain phi_i} and Lemma~\ref{expected lb psi}.

\begin{theorem} 
	When Algorithm~\ref{unconstrained drg} terminates, it returns a feasible set $S_n$ with
	\begin{equation*}
		\Ebb_{S_n} [ g(S_n) - \ell(S_n) ] \geq \left( 1 -\frac 1n \right)^{n-1} g(OPT) - \ell(OPT) \geq e^{-1} g(OPT) - \ell(OPT),
	\end{equation*}
	where $OPT \in \argmax \{ g(A) - \ell(A) : A \subseteq \Omega \}$. And, during the iteration, the total value oracle queries are $O(n)$.
\end{theorem}

\section{Conclusion} \label{sec:conclusion}
In this paper, we study the regularized non-monotone submodular maximization problem thoroughly. We propose several algorithms for the optimization problem subject to various constraints, including matroid constraint, cardinality constraint and no constraint. We give a systematical and unified analysis of the performance guarantee and the complexity of value oracle query of those algorithms we propose. According to the analysis, our algorithms are both effective and efficient.

\section*{Ackowledgement}
This research is supported by the National Natural Science Foundation of China under Grant Numbers 11991022 and 12071459.

%\bibliography{wpref}

\begin{thebibliography}{}
	
	\bibitem[Ageev et~al., 2001]{ageev20010}
	Ageev, A., Hassin, R., and Sviridenko, M. (2001).
	\newblock A 0.5-approximation algorithm for max dicut with given sizes of
	parts.
	\newblock {\em SIAM Journal on Discrete Mathematics}, 14(2):246--255.
	
	\bibitem[Alon and Spencer, 2004]{alon2004probabilistic}
	Alon, N. and Spencer, J.~H. (2004).
	\newblock {\em The probabilistic method}.
	\newblock John Wiley \& Sons.
	
	\bibitem[Buchbinder et~al., 2014]{buchbinder2014submodular}
	Buchbinder, N., Feldman, M., Naor, J., and Schwartz, R. (2014).
	\newblock Submodular maximization with cardinality constraints.
	\newblock In {\em Proceedings of the twenty-fifth annual ACM-SIAM symposium on
		Discrete algorithms}, pages 1433--1452. SIAM.
	
	\bibitem[Buchbinder et~al., 2017]{buchbinder2017comparing}
	Buchbinder, N., Feldman, M., and Schwartz, R. (2017).
	\newblock Comparing apples and oranges: Query trade-off in submodular
	maximization.
	\newblock {\em Mathematics of Operations Research}, 42(2):308--329.
	
	\bibitem[Buchbinder et~al., 2015]{buchbinder2015tight}
	Buchbinder, N., Feldman, M., Seffi, J., and Schwartz, R. (2015).
	\newblock A tight linear time (1/2)-approximation for unconstrained submodular
	maximization.
	\newblock {\em SIAM Journal on Computing}, 44(5):1384--1402.
	
	\bibitem[Calinescu et~al., 2011]{calinescu2011maximizing}
	Calinescu, G., Chekuri, C., Pal, M., and Vondr{\'a}k, J. (2011).
	\newblock Maximizing a monotone submodular function subject to a matroid
	constraint.
	\newblock {\em SIAM Journal on Computing}, 40(6):1740--1766.
	
	\bibitem[Feige, 1998]{feige1998threshold}
	Feige, U. (1998).
	\newblock A threshold of ln n for approximating set cover.
	\newblock {\em Journal of the ACM (JACM)}, 45(4):634--652.
	
	\bibitem[Feige et~al., 2011]{feige2011maximizing}
	Feige, U., Mirrokni, V.~S., and Vondr{\'a}k, J. (2011).
	\newblock Maximizing non-monotone submodular functions.
	\newblock {\em SIAM Journal on Computing}, 40(4):1133--1153.
	
	\bibitem[Feldman, 2020]{feldman2020guess}
	Feldman, M. (2020).
	\newblock Guess free maximization of submodular and linear sums.
	\newblock {\em Algorithmica}, pages 1--26.
	
	\bibitem[Feldman et~al., 2011]{feldman2011unified}
	Feldman, M., Naor, J., and Schwartz, R. (2011).
	\newblock A unified continuous greedy algorithm for submodular maximization.
	\newblock In {\em 2011 IEEE 52nd Annual Symposium on Foundations of Computer
		Science}, pages 570--579. IEEE.
	
	\bibitem[Frieze and Jerrum, 1997]{frieze1997improved}
	Frieze, A. and Jerrum, M. (1997).
	\newblock Improved approximation algorithms for maxk-cut and max bisection.
	\newblock {\em Algorithmica}, 18(1):67--81.
	
	\bibitem[Gharan and Vondr{\'a}k, 2011]{gharan2011submodular}
	Gharan, S.~O. and Vondr{\'a}k, J. (2011).
	\newblock Submodular maximization by simulated annealing.
	\newblock In {\em Proceedings of the twenty-second annual ACM-SIAM symposium on
		Discrete Algorithms}, pages 1098--1116. SIAM.
	
	\bibitem[Gr{\"o}tschel et~al., 1981]{grotschel1981ellipsoid}
	Gr{\"o}tschel, M., Lov{\'a}sz, L., and Schrijver, A. (1981).
	\newblock The ellipsoid method and its consequences in combinatorial
	optimization.
	\newblock {\em Combinatorica}, 1(2):169--197.
	
	\bibitem[Harshaw et~al., 2019]{harshaw2019submodular}
	Harshaw, C., Feldman, M., Ward, J., and Karbasi, A. (2019).
	\newblock Submodular maximization beyond non-negativity: Guarantees, fast
	algorithms, and applications.
	\newblock {\em arXiv preprint arXiv:1904.09354}.
	
	\bibitem[Iwata et~al., 2001]{iwata2001combinatorial}
	Iwata, S., Fleischer, L., and Fujishige, S. (2001).
	\newblock A combinatorial strongly polynomial algorithm for minimizing
	submodular functions.
	\newblock {\em Journal of the ACM (JACM)}, 48(4):761--777.
	
	\bibitem[Kazemi et~al., 2020]{kazemi2020regularized}
	Kazemi, E., Minaee, S., Feldman, M., and Karbasi, A. (2020).
	\newblock Regularized submodular maximization at scale.
	\newblock {\em arXiv preprint arXiv:2002.03503}.
	
	\bibitem[Kempe et~al., 2003]{kempe2003maximizing}
	Kempe, D., Kleinberg, J., and Tardos, {\'E}. (2003).
	\newblock Maximizing the spread of influence through a social network.
	\newblock In {\em Proceedings of the ninth ACM SIGKDD international conference
		on Knowledge discovery and data mining}, pages 137--146.
	
	\bibitem[Khuller et~al., 1999]{khuller1999budgeted}
	Khuller, S., Moss, A., and Naor, J.~S. (1999).
	\newblock The budgeted maximum coverage problem.
	\newblock {\em Information processing letters}, 70(1):39--45.
	
	\bibitem[Krause and Guestrin, 2007]{krause2007near}
	Krause, A. and Guestrin, C. (2007).
	\newblock Near-optimal observation selection using submodular functions.
	\newblock In {\em AAAI}, volume~7, pages 1650--1654.
	
	\bibitem[Krause and Guestrin, 2012]{krause2012near}
	Krause, A. and Guestrin, C.~E. (2012).
	\newblock Near-optimal nonmyopic value of information in graphical models.
	\newblock {\em arXiv preprint arXiv:1207.1394}.
	
	\bibitem[Krause et~al., 2008]{krause2008near}
	Krause, A., Singh, A., and Guestrin, C. (2008).
	\newblock Near-optimal sensor placements in gaussian processes: Theory,
	efficient algorithms and empirical studies.
	\newblock {\em Journal of Machine Learning Research}, 9(Feb):235--284.
	
	\bibitem[Lin and Bilmes, 2010]{lin2010multi}
	Lin, H. and Bilmes, J. (2010).
	\newblock Multi-document summarization via budgeted maximization of submodular
	functions.
	\newblock In {\em Human Language Technologies: The 2010 Annual Conference of
		the North American Chapter of the Association for Computational Linguistics},
	pages 912--920.
	
	\bibitem[Nemhauser and Wolsey, 1978]{nemhauser1978best}
	Nemhauser, G.~L. and Wolsey, L.~A. (1978).
	\newblock Best algorithms for approximating the maximum of a submodular set
	function.
	\newblock {\em Mathematics of operations research}, 3(3):177--188.
	
	\bibitem[Nemhauser et~al., 1978]{nemhauser1978analysis}
	Nemhauser, G.~L., Wolsey, L.~A., and Fisher, M.~L. (1978).
	\newblock An analysis of approximations for maximizing submodular set
	functions—i.
	\newblock {\em Mathematical programming}, 14(1):265--294.
	
	\bibitem[Schrijver, 2000]{schrijver2000combinatorial}
	Schrijver, A. (2000).
	\newblock A combinatorial algorithm minimizing submodular functions in strongly
	polynomial time.
	\newblock {\em Journal of Combinatorial Theory, Series B}, 80(2):346--355.
	
	\bibitem[Sviridenko et~al., 2017]{sviridenko2017optimal}
	Sviridenko, M., Vondr{\'a}k, J., and Ward, J. (2017).
	\newblock Optimal approximation for submodular and supermodular optimization
	with bounded curvature.
	\newblock {\em Mathematics of Operations Research}, 42(4):1197--1218.
	
	\bibitem[Vondr{\'a}k, 2013]{vondrak2013symmetry}
	Vondr{\'a}k, J. (2013).
	\newblock Symmetry and approximability of submodular maximization problems.
	\newblock {\em SIAM Journal on Computing}, 42(1):265--304.
	
\end{thebibliography}

\end{document}